\documentclass[journal]{IEEEtran}
\ifCLASSINFOpdf
\else
\fi

\usepackage{cite}
\usepackage{amsmath,amssymb,amsfonts}
\usepackage[noend]{algpseudocode}
\usepackage{algorithmicx,algorithm}
\usepackage{graphicx}
\usepackage{textcomp}
\usepackage{subfigure}
\usepackage{xcolor}
\usepackage{amsthm}
\usepackage{float}
\usepackage{tabularx}
\usepackage{booktabs}
\usepackage{svg}
\usepackage{algorithm}
\usepackage{algpseudocode}
\usepackage[algo2e]{algorithm2e}
\makeatletter
\newcommand{\removelatexerror}{\let\@latex@error\@gobble}
\makeatother
\newtheorem{definition}{Definition}
\newtheorem{proposition}{Proposition}
\newtheorem{assumption}{Assumption}

\newtheorem{theorem}{Theorem}
\newtheorem{lemma}{Lemma}
\let\oldnl\nl
\newcommand{\nonl}{\renewcommand{\nl}{\let\nl\oldnl}}

\begin{document}
%
\title{Alliance Makes Difference? Maximizing Social Welfare in Cross-Silo Federated Learning}
%
%
%

\author{Jianan~Chen,
        Qin~Hu,
        and~Honglu~Jiang
\thanks{Jianan Chen and Qin Hu (Corresponding Author) are with the Department of Computer and Information Science, Indiana University - Purdue University Indianapolis, IN, USA. Email: jc144@iu.edu, qinhu@iu.edu}       
\thanks{Honglu Jiang is with the Department of Computer Science and Software Engineering, Miami University, Oxford, OH, USA. Email: jiangh34@miamioh.edu}
\thanks{This work is partly supported by the US NSF under grant CNS-2105004.}
}

%
%

\markboth{Journal of \LaTeX\ Class Files,~Vol.~14, No.~8, August~2015}%
{Shell \MakeLowercase{\textit{et al.}}: Bare Demo of IEEEtran.cls for IEEE Journals}
%



\maketitle

\begin{abstract}
As one of the typical settings of Federated Learning (FL), cross-silo FL allows organizations to jointly train an optimal Machine Learning (ML) model. In this case, some organizations may try to obtain the global model without contributing their local training power, lowering the social welfare. In this paper, we model the interactions among organizations in cross-silo FL as a public goods game and theoretically prove that there exists a social dilemma where the maximum social welfare is not achieved in Nash equilibrium. To overcome this dilemma, we employ the Multi-player Multi-action Zero-Determinant (MMZD) strategy to maximize the social welfare. With the help of the MMZD, an individual organization can unilaterally control the social welfare without extra cost. Since the MMZD strategy can be adopted by all organizations, we further study the case of multiple organizations jointly adopting the MMZD strategy to form an MMZD Alliance (MMZDA). We prove that the MMZDA strategy can strengthen the control of the maximum social welfare. Experimental results validate that the MMZD strategy is effective in obtaining the maximum social welfare and the MMZDA can achieve a larger maximum value.
\end{abstract}


%
\IEEEpeerreviewmaketitle

\section{Introduction}
%
%
%
%
\IEEEPARstart{I}{n} Federated Learning (FL), clients cooperatively train a Machine Learning (ML) model with their decentralized datasets under the coordination of a central server \cite{mcmahan2017communication}. One of the typical settings of FL is cross-silo FL \cite{kairouz2019advances} where a neutral third-party agent acts as the central server and clients are a group of organizations, aiming to jointly train an optimal ML model for their respective use. In this case, these organizations are also the owners of the global model and can utilize the well-trained global model to further process tasks for their own interests. 

An optimal global model with high performance requires the organizations in cross-silo FL to collaborate efficiently so as to bring considerable benefits to all participants. The cooperative behavior of organizations participating in global aggregation can improve the social welfare. In fact, there are many studies on optimizing the social welfare in cross-silo FL by improving the model performance \cite{huang2021personalized,wang2021efficient,nandury2021cross,majeed2020cross}, increasing the convergence speed \cite{marfoq2020throughput}, reducing the communication cost \cite{zhang2020batchcrypt}, protecting privacy \cite{heikkila2020differentially,li2021practical,long2021federated} and security \cite{jiang2021flashe}, etc.

However, since every organization in cross-silo FL can obtain the final global model regardless of its contribution, the well-trained model becomes a public good, which is non-excludable and non-rivalrous for all organizations \cite{tang2021incentive}. This leads to selfish behaviors that some organizations may only consider their own interests via inactively participating in local training to obtain the final global model for free or at a lower cost. The spread of this behavior can result in a huge loss of the social welfare, and then none of the organizations can get the optimal model, which compromises the long-term stability and sustainability of cross-silo FL.

Most of the existing studies improve the social welfare by designing incentive mechanisms to promote organizations’ full cooperation in cross-silo FL \cite{tang2021incentive,li2022incentive,zeng2022incentive,zhang2022enabling,das2021cross}. However, incentive mechanism requires extra negotiation costs since organizations need to reach a consensus on the mechanism in advance. And it also demands additional running costs, where a distributed algorithm runs over all organizations, clearly adding more burden to organizations.

In this paper, we model the interactions among organizations in cross-silo FL into a public goods game, named cross-silo FL game. Instead of incentivizing organizations' participation, we take a brand-new approach using the Multi-player Multi-action Zero-Determinant (MMZD) strategy \cite{7381622} to directly maximize the social welfare in cross-silo FL without causing additional negotiation costs and running costs for all organizations. Moreover, the MMZD Alliance (MMZDA) formed by multiple MMZD players is studied to explore whether a better maximum value of the social welfare can be achieved. Another outstanding advantage of our methods is that they can be applied to any cross-silo FL scenario no matter what strategies or actions other organizations perform. 

In summary, our contributions include (a preliminary version of this paper is presented in ICASSP 2022 \cite{chen2022social}):
\begin{itemize}
\item  We model the interactions among organizations in cross-silo FL as a public goods game, focusing on the organization's strategy rather than designing an extra mechanism to solve the social welfare maximization problem.
\item We reveal the existence of the social dilemma in cross-silo FL by mathematical proof for the first time, which demonstrates the adverse effect of selfish behaviors in cross-silo FL in the view of game theory. This can be used as a theoretical basis for exploring organizations' behaviors in cross-silo FL.
\item We overcome the social dilemma by employing the MMZD strategy from the perspectives of individual organization and alliance. Specifically, any organization can unilaterally maximize the social welfare, which ensures the social welfare in cross-silo FL at a certain level and maintains the stability of the system. 
\item We further study the scenario in which multiple organizations adopt the same MMZD strategy, forming the MMZDA. We theoretically prove that the maximum social welfare controlled by the MMZDA can reach a higher value. This approach also extends the applications of the MMZD.
\item Experiments prove the effectiveness of the MMZD strategy in maximizing the social welfare. And the maximum value of the social welfare is able to be enlarged by the MMZDA strategy. 
\end{itemize}

The rest of the paper is organized as follows. The related work about cross-silo FL is summarized in Section II. In Section III, we formulate the cross-silo FL game to model the interactions among organizations in cross-silo FL, and further discover the social dilemma in the game. We propose a method based on MMZD strategy for individual organization to control the social welfare in Section IV. Section V studies the MMZD strategy employed by multiple organizations, namely MMZDA, and validates that the MMZDA can enlarge the maximum value of the social welfare compared to the individual MMZD. Simulation results are reported in Section VI, followed by the conclusion in Section VII.

\section{Related Work}

Existing works related to cross-silo FL can be classified into three categories: the optimization of aggregation algorithm, security and privacy protection, and the incentive mechanism design.

Aggregation algorithms are developed to enhance the performance of cross-silo FL. Based on the original FedAvg algorithm \cite{mcmahan2017communication}, various algorithms to improve the convergence speed, accuracy, and security were proposed in cross-silo FL setting. Marfoq et al. introduced practical algorithms to design an averaging policy under a decentralized model for achieving the fastest convergence \cite{marfoq2020throughput}. And Huang et al. proposed FedAMP to overcome non-iid challenges \cite{huang2021personalized}. Zhang et al. reduced the encryption and communication overhead caused by additively homomorphic encryption, which lowered the cost of aggregation as well \cite{zhang2020batchcrypt}.

Security and privacy protection issues also received attention in cross-silo FL. Heikkil et al. combined additively homomorphic secure summation protocols with differential privacy to guarantee strict privacy for individual data subjects in the cross-silo FL setting \cite{heikkila2020differentially}. Li et al. proposed a brand-new one-shot algorithm that can flexibly achieve differential privacy guarantees \cite{li2020practical}. Jiang et al. designed FLASHE, an optimized homomorphic encryption scheme to meet the requirements of semantic security and additive homomorphism in cross-silo FL \cite{jiang2021flashe}. Chu et al. proposed a federated estimation method to accurately estimate the fairness of a model, namely avoiding the bias of data, without infringing the data privacy of any party \cite{chu2021fedfair}.

Many incentive mechanisms based on auction \cite{le2021incentive}, contract \cite{tian2021contract,feng2019joint}, and pricing \cite{shao2019multimedia} were proposed for FL since the performance of FL is affected by clients' behaviors. Unfortunately, most of schemes cannot be adapted to the cross-silo FL directly. Only a few studies successfully conducted incentive mechanisms in cross-silo FL to encourage organizations to participate in global aggregation. Tang et al. proposed an incentive mechanism for cross-silo FL to address the organization heterogeneity and the public goods characteristics by solving a non-convex optimization problem\cite{tang2021incentive}. Li et al. proposed an incentive mechanism for the cross-silo FL scenario, stimulating the organizations to provide more high-quality data by deploying a knowledge distillation algorithm \cite{li2022incentive}. Zeng et al. considered cross-silo FL as a perfectly competitive market, proposing the Cournot model,
Stackelberg-Cournot model, and Cournot-Stackelberg model to analyze the economical behaviors of organizations\cite{zeng2022incentive}. Zhang et al. explored the long-term participation of organizations in cross-silo FL by calculating their equilibrium participation strategy, where an algorithm was designed for organizations to reduce free riders and increase the amount of local data for global aggregation \cite{zhang2022enabling}. 

However, the application of incentive mechanisms in cross-silo FL scenario still faces many challenges. First of all, most of the existing incentive mechanisms focus on encouraging more organizations to participate in cross-silo FL, instead of improving the performance of the global model from the perspective of the social welfare. Secondly, the global model in cross-silo FL has the non-exclusive nature of public goods, leading to potential free-riding behaviors, which is rarely considered in the existing research. Last but not least, current incentive mechanisms usually require participants or servers to spend additional computing resources. As a result, complex designs and additional computing costs can bring a burden on the organizations. 


In light of the above analysis, our work is distinct with the existing approaches in the following aspects. First, our study reveals the social dilemma in the cross-silo FL. Second, we are committed to directly maximizing the social welfare to overcome the social dilemma. Moreover, we adopt the MMZD strategy without additional cost to control the maximum value of the social welfare. Furthermore, we use the MMZDA to further expand the ability to control the maximum value of the social welfare. 
\section{System Model} \label{Modeling}
We consider a cross-silo FL scenario with a set of organizations, denoted as $\mathcal{N}=\{1,2,...,N\}$. All organizations rely on a central server to collaboratively conduct global model training for a specific task, where each of them has their own data for local training\footnote{The data distributions of organizations would not affect the problem formulation and results of our method. Here we assume that the data of organizations are i.i.d., following other works on cross-silo FL \cite{heikkila2020differentially,wang2022safeguarding}.}. The goal of organizations is to obtain an optimal global model, minimizing the loss based on all datasets. In one communication round, the central server collects the results of local model updates from all organizations, aggregates to obtain the global model, and then distributes it to everyone for the next round of local training. 

In each round of local training, every organization performs $K$ iterations of model training. We denote the number of global communication rounds for aggregation as $r$. For the current task, the action of organization $i\in\mathcal{N}$, denoted as $y_i\in\{0,1,...,r\}$, represents the number of communication rounds it participates in the task. The organization $i$ randomly selects $y_i$ communication rounds to participate. Then, $\mathbf{y}=(y_1,...,y_i,...,y_N)$ denotes the action vector of all organizations. Here we assume that all organizations in this cross-silo FL may participate in fewer global aggregations due to laziness or selfishness, but they do not carry out malicious attacks, such as model poisoning attack.

According to the cross-silo FL model, all organizations get the same model in return. Inspired by \cite{tang2021incentive}, we define the revenue of organization $i$ as:
\begin{equation}\label{eq:revenue}
\Phi_i(\mathbf{y})=m_i(\chi_0-\chi(\mathbf{y})),
\end{equation}
where $m_i$ (in dollars per unit of precision function) denotes the unit revenue of organization $i$ by using the returned final model, $\chi_0$ denotes the precision of the untrained model, and $\chi(\mathbf{y})$ denotes the precision of the trained global model after the actions of organizations in the action vector $\mathbf{y}$. Specifically, $\chi(\mathbf{y})$ can be modeled as:
\begin{equation}
\chi(\mathbf{y})=\frac{\theta_0}{\theta_1+K \sum_{i\in N} y_i},
\end{equation}
with positive coefficients $\theta_0$ and $\theta_1$ \cite{li2019convergence} being derived based on the loss function, neural network, and local datasets. In particular, we have $\chi_0=\frac{\theta_0}{\theta_1}$. The revenue of each organization is proportional to the difference between the expected loss after $r$ communication rounds aggregation (i.e., $\chi(\mathbf{y})$) and the minimum expected loss (i.e., $\chi_0$) \cite{tang2021incentive}. As the number of total participation rounds increases, the marginal decrease of the difference reduces.

We define the cost of organization $i$ as:
\begin{equation}
\Psi_i(y_i)=C^{i}_p(y_i)+C^{i}_m(y_i).    
\end{equation}
The cost is composed of the organization's computation cost $C^{i}_p(y_i)$ and its communication cost $C^{i}_m(y_i)$. The computation cost $C^{i}_p(y_i)=\beta_i K y_i$, where $\beta_i$ is a positive parameter, denoting the computation cost of each iteration in organization $i$'s local training\footnote{As \cite{tran2019federated} shows, $\beta_i=\frac{\alpha_i}{2} f^2_i d_i S_i$, where $\frac{\alpha_i}{2}$ is the effective capacitance coefficient of organization $i$'s computing chipset, $f_i$ denotes the calculation processing capacity, $d_i$ denotes the number of data units, and $S_i$ denotes the number of CPU cycles required by organization $i$ to process one data unit.}. While the communication cost is defined as $C^{i}_m(y_i)=\rho_i\tau_i y_i$, where $\rho_i$ is a positive parameter denoting the communication power of organization $i$, and $\tau_i$ is its time cost uploading the model updates in one global communication round\footnote{According to \cite{tran2019federated}, the time cost $\tau_i$ of organization $i$ is calculated by $\tau_i=\frac{D_i}{R_i}$, where $D_i$ denotes the size of model updates and $R_i$ denotes the transmission rate. The transmission rate $R_i$ can be calculated as $R_i=B\ln{(1+\frac{G_i\rho_i}{N_0})}$, where $B$ is the bandwidth, $N_0$ is the background noise, and $G_i$ is the channel gain.}.

Then the utility of organization $i$ is defined as
the difference between its revenue and cost:
\begin{equation}\label{eq:utility}
U^i(\mathbf{y})=\Phi_i(\mathbf{y})-\Psi_i(y_i).
\end{equation}
According to previous statements, we model the interactions among organizations as a $\emph{cross-silo FL game}$. 
\begin{definition}
(Cross-silo FL game). In the cross-silo FL game, the participating organizations act as players, where organization $i$'s action and utility are $y_i$ and $U^i(\mathbf{y})$, respectively.
\end{definition}
The cross-silo FL game can be iterative since these organizations in cross-silo FL usually cooperate for a long time to finish multiple FL tasks. Each game round in the cross-silo FL game corresponds to a certain FL task. Moreover, the social welfare in the cross-silo FL game can be denoted as the total utility of all organizations $i$, namely $\sum^N_{i=1} U^i(\mathbf{y})$. In the cross-silo FL game, we find that the social dilemma occurs if $\Phi_i(\mathbf{y})-C^{i}_p(y_i)<0$, which can be summarized as below.

\begin{assumption}\label{assumption}
In the cross-silo FL game, we assume that if any organization $i\in\mathcal{N}$ trains the local model with only its own dataset without joining cross-silo FL, the utility is negative. Namely, $m_i(\chi_0-\frac{\theta_0}{\theta_1+Ky_i})-C^{i}_p(y_i)<0$ holds for all $i\in\mathcal{N}$.
\end{assumption}
The first part of the formula $m_i(\chi_0-\frac{\theta_0}{\theta_1+Ky_i})$ represents the model gain that organization $i$ can obtain from local training. The second part of the formula $C^{i}_p(y_i)$ represents the local computation cost of organization $i$. In fact, this condition is consistent with the organizations' motivation to participate in global aggregation in the cross-silo FL game.
\begin{lemma}\label{lemma:bestresponse}
(Nash equilibrium). Under Assumption \ref{assumption}, a Nash equilibrium point in the cross-silo FL game is:
\begin{equation*}
    y^{NE}_i=0.
\end{equation*}
\begin{proof}
Referring to (\ref{eq:utility}), we can derive the derivative of $U^i$ regarding $y_i$ as $\frac{\partial U^i}{\partial y_i}=m_i\frac{K \theta_0}{(\theta_1+K \sum y_i)^2}- (\beta_i K+a_i\tau_i)$. Given $m_i(\chi_0-\frac{\theta_0}{\theta_1+Ky_i})-C^{i}_p(y_i)<0$, we have:
\begin{equation*}
m_i\frac{K\theta_0}{(\theta_1+K y_i)\theta_1}< \beta_i K,
\end{equation*}
which leads to:
\begin{equation*}
m_i\frac{K \theta_0}{(\theta_1+K \sum y_i )^2}<m_i\frac{K\theta_0}{(\theta_1+K y_i)\theta_1}<\beta_i K<\beta_i K+a_i\tau_i.
\end{equation*}
Then we can draw the conclusion that $\frac{\partial U^i}{\partial y_i}<0$. Thus, the utility function of organization $i$ decreases monotonically with $y_i$. So the best response of organization $i$ is $y^{BR}_i=0$ regardless of other organizations' strategies. And thus, $\mathbf{y}^{NE}=(0,0,...,0)$ is a Nash equilibrium point.
\end{proof}
\end{lemma}
In addition, the general solution of the Nash equilibrium in the cross-silo FL game can be expressed as:
\begin{equation*}\label{eq:NashEq}
\begin{split}
\begin{cases}
y_1=\arg \max{[\Phi_1(y_1,\textbf{y}_{-1})-\Psi_1(y_1)]},\\
\cdots\\
y_N=\arg \max{[\Phi_N(y_N,\textbf{y}_{-N})-\Psi_N(y_N)]},\\
\end{cases}
\end{split}
\end{equation*}
where $\textbf{y}_{-i}$ denotes the actions other than organization $i$. In this paper, we concentrate on the Nash equilibrium point identified by Lemma \ref{lemma:bestresponse} as it would lead to a social dilemma.

\begin{definition}\label{def:SD}
(Social dilemma).  In the cross-silo FL game, the social dilemma occurs when the best actions of individual organizations contrast with the maximization of social welfare, i.e., a Nash equilibrium point is not the social welfare maximum point.
\end{definition}
\begin{theorem}\label{th:SD}
Under Assumption \ref{assumption} and Definition \ref{def:SD}, there exists a social dilemma in the cross-silo FL game. 

\begin{proof}
According to Lemma \ref{lemma:bestresponse}, $\mathbf{y}^{NE}=(0,0,\dots,0)$ is a Nash equilibrium point. Next, we prove that the point $\mathbf{y}^{r}=(r,r,\dots,r)$ results in the social welfare:
\begin{equation*}
\sum^N_{i=1}U^i(\mathbf{y}^{r})=\sum \Psi_i(\mathbf{y}^{r})-\sum C^i_p-\sum C^i_m>0,
\end{equation*}
which is higher than that in the Nash equilibrium point:
\begin{equation*}
\sum^N_{i=1}U^i(\mathbf{y}^{NE})=-\sum C^i_m<0.    
\end{equation*}
The Nash equilibrium point cannot lead to the maximum social welfare, so the social dilemma exists.
\end{proof}
\end{theorem}

If the organization only pursues its own interest and does not participate in the communication round for global aggregation, it will lead to a low social welfare. From the perspective of an individual organization, free riding is an attractive action to get the final global model at a low cost. However, when all organizations are free riders, the model performance can be poor, which is unfavorable for everyone. Therefore, an organization cannot arbitrarily choose the free riding behavior for its own benefit without considering the social welfare. From the perspective of collective, if all organizations only look out for their own interests, the model performance of cross-silo FL systems will degrade as selfish behavior spreads.  

From the above analysis, both individual and group have sufficient motivation to pay attention to the social welfare and try to maximize the social welfare. For reference, we summarize key notations used in the system model in Table I. 

\begin{table}[htbp]\label{notation}
\centering
\caption{Key Notations.}
\centering
\begin{tabular}{|c|l|}
\hline
     Notation & Meaning \\ \hline
     $N$   &The number of organizations \\ \hline
     $K$   &The number of local training iterations for every organization\\ \hline
      $r$   &The number of global communication rounds for aggregation \\ \hline
      $y_i$   & \begin{tabular}[c]{@{}l@{}}The number of communication rounds organization $i$\\ participates in the current task\end{tabular}\\ \hline
      $\mathbf{y}$   &The action vector of all organizations\\ \hline
      $\Phi_i$   &The revenue of organization $i$\\ \hline
      $m_i$   &The unit revenue of organization $i$ by using the final model\\ \hline
      $\chi(\mathbf{y})$   &\begin{tabular}[c]{@{}l@{}} The precision of the trained global model with the\\corresponding action vector $\mathbf{y}$\end{tabular}\\ \hline
      $\Psi_i$   &The cost of organization $i$\\ \hline
      $C^{i}_p(y_i)$   &The computation cost of organization $i$\\ \hline
      $C^{i}_m(y_i)$   &The communication cost of organization $i$\\ \hline
      $\beta_i$   & \begin{tabular}[c]{@{}l@{}}The computation cost of each iteration in organization $i$'s\\ local training\end{tabular}\\ \hline
      $\rho_i$   & The communication power of organization $i$ \\ \hline 
      $\tau_i$   & \begin{tabular}[c]{@{}l@{}} The time cost of organization $i$ uploading the model\\ updates in one global communication round\end{tabular} \\ \hline
      $U^i(\mathbf{y})$   & \begin{tabular}[c]{@{}l@{}} The utility of organization $i$ with the corresponding\\ action vector $\mathbf{y}$\end{tabular}\\ \hline
\end{tabular}
\end{table}

\section{Social welfare maximization by MMZD}\label{maxwelfare}
According to the analysis above, we can see that the underlying cause of the social dilemma is selfishness, leading to the loss of all organizations, namely the low social welfare. All organizations are reluctant to face the bad consequences of low social welfare, so in this section, we explore whether an organization can maximize the social welfare as much as possible in the cross-silo FL game. We start by analyzing an organization's strategy.

In each game round, any organization can choose the action $y_i\in\{0,1,...,r\}$, so there are $(r+1)^N$ possible outcomes for each game round. Fig. \ref{fig:illustration} describes an example of the cross-silo FL game with two organizations and three actions, i.e., $N=2$ and $r=2$, in which all possible outcomes can be denoted as $(y_1,y_2)\in\{(0,0),(0,1),(0,2),(1,0),(1,1),(1,2),(2,0),(2,1),(2,2)\}$. For arbitrary organization $i\in\mathcal{N}$, its mixed strategy $\mathbf{p}^i$ is defined as:
\begin{equation}
\mathbf{p}^i=[p^i_{1,0},p^i_{1,1},...,p^i_{1,r},p^i_{2,0},...,p^i_{j,g},...,p^i_{(r+1)^{N},r}]^T,
\end{equation}
where $p^i_{j,g} (j\in\{1,2,...,(r+1)^N\},g\in\{0,1,...,r\})$ represents the probability of organization $i$ choosing action $y_i=g$ in the current game round and other organizations choosing the same actions as $j$-th outcome of the previous game round. Here we assume that the organizations have one-round memory since a long-memory player has no priority against others with short memory \cite{7381622}. Under this assumption, the decision of the current game round of organization $1$ is only related to the previous game round. Therefore, this process has the Markov property, and we can regard it as a Markov process. As presented in Fig. \ref{fig:illustration}, for the previous outcome $(y_1,y_2)=(0,2)$, the conditional probability of organization $1$ to select action $y_1=1$ is $p^1_{3,1}$ and the conditional probability of organization 2 to select action $y_2=0$ is $p^2_{3,0}$. In addition, the corresponding utility vector $\mathbf{u}^i$ is denoted as:
\begin{equation}
\mathbf{u}^i=[u^i_{1,0},u^i_{1,1},...,u^i_{1,r},u^i_{2,0},...,u^i_{j,g},...,u^i_{(r+1)^{N},r}]^T,
\end{equation}
where each utility $u^i_{j,g}$ of organization $i$ choosing action $y_i=g$ in the $j$-th outcome can be calculated by $u^i_{j,g}=U^i(\mathbf{y}^{(j,g)})$, with $\mathbf{y}^{(j,g)}$ denoting the action vector $\mathbf{y}$ corresponding to the $j$-th outcome but $y_i=g$. Fig. \ref{fig:example}(a) presents the strategy vectors and utility vectors based on the example shown in Fig. \ref{fig:illustration}.

\begin{figure}
\centering
\includegraphics[scale=0.7]{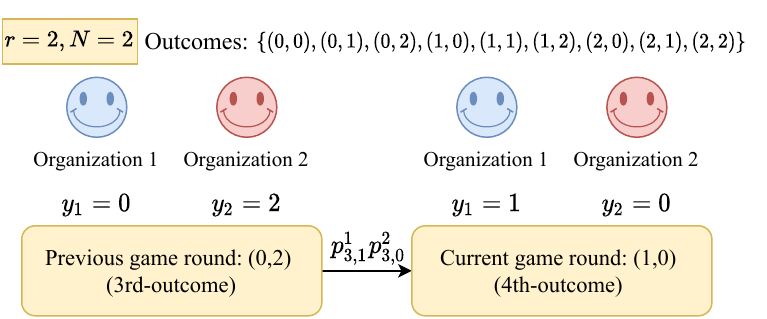}
\caption{Illustration of the cross-silo FL game with two organizations and three actions ($N=2$ and $r=2$).}
\label{fig:illustration} 
\end{figure}

In the cross-silo FL model, an organization’s current move depends only on its last action and the action vector $\mathbf{y}$ in the last game round. We can construct a Markov matrix $\mathbf{M}=[M_{vw}]_{(r+1)^{N} \times (r+1)^{N}}$, with each element $M_{vw}$ denoting the one-step transition probability from state $v$ to $w$. Fig. \ref{fig:example}(b) shows the Markov matrix $\mathbf{M}$ for the case of $r=2$ and $N=2$. For example, the element $p^1_{3,1}p^2_{3,0}$ is at the 3rd row and 4th column in Fig. \ref{fig:example}, which represents the possibility of transitioning from the 3rd-outcome (0,2) in the previous round to the 4th-outcome (1,0) in the current round. Then, we define $\mathbf{M}'\equiv \mathbf{M}-\textbf{I}$ where $\textbf{I}$ is an identity matrix. And we assume the stationary vector of $\mathbf{M}$ is $\pi$. Given $\pi^T\mathbf{M}=\pi^T$, we can draw that $\pi^T\mathbf{M}'=0$. According to Cramer’s rule, we have: 
\begin{equation}
Adj(\mathbf{M}')\mathbf{M}'=\det(\mathbf{M}')\textbf{I}=0,
\end{equation}
where $Adj(\mathbf{M}')$ denotes the adjugate matrix of $\mathbf{M}'$. Thus, it can be noted that every row of $Adj(\mathbf{M}')$ is proportional to $\pi$. Hence for any vector $\mathbf{a}=(a_1,a_2,\dots,a_{(r+1)^N})^T$, we can draw $\pi^T\cdot\mathbf{a}=\det (\mathbf{p}^1,\dots,\mathbf{p}^N,\mathbf{a})$ according to \cite{press2012iterated}. 

\begin{figure*}
\includegraphics[scale=0.75]{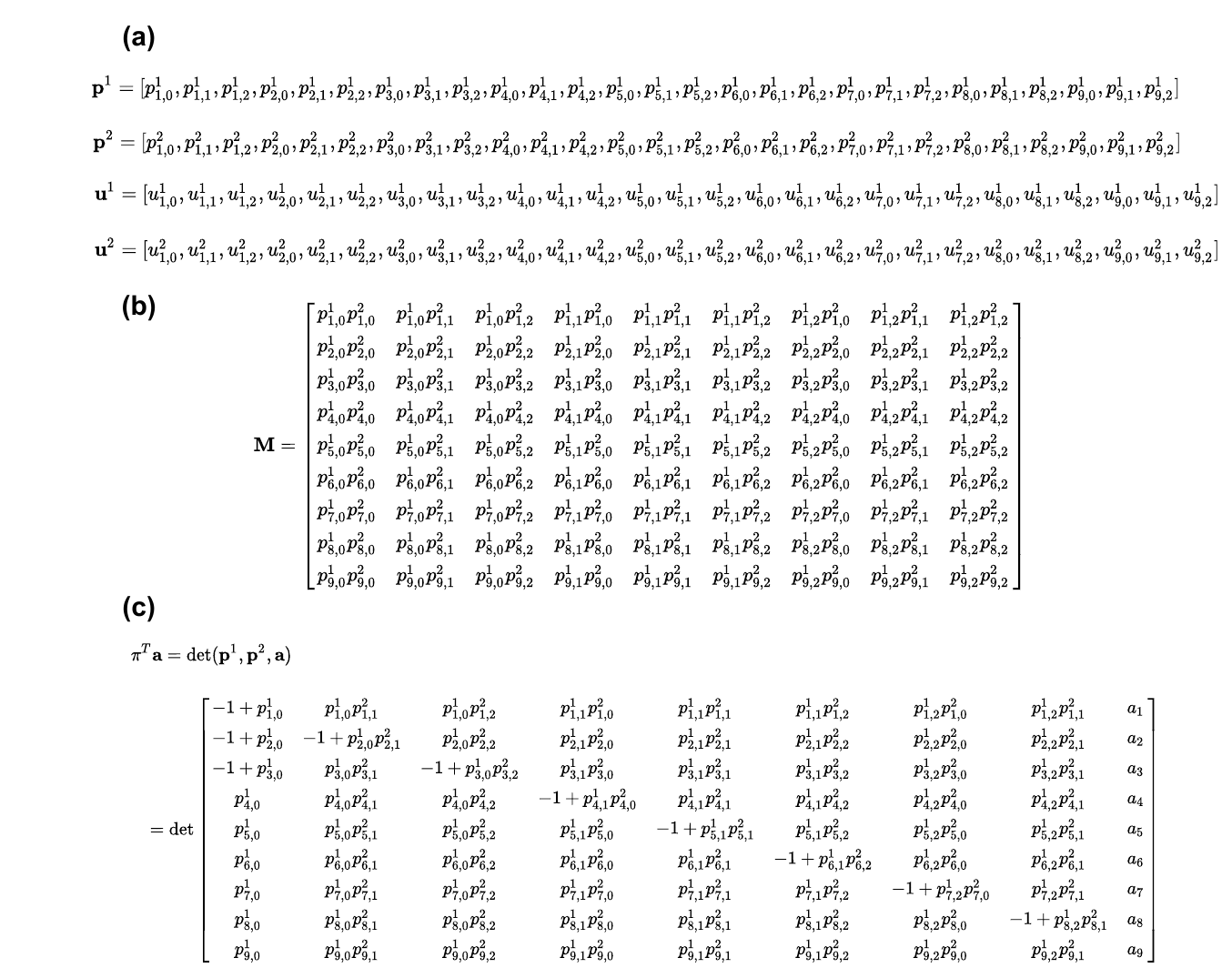}
\caption{The strategy vectors, utility vectors, Markov matrix, and determinant of $\pi^T\cdot\mathbf{a}$ after elementary transformations in the cross-silo FL game example with two organizations and three actions shown in Fig. \ref{fig:illustration}. (a) The strategy vectors and the corresponding utility vectors of the cross-silo FL game. (b) The Markov matrix $\mathbf{M}$ of the cross-silo FL game. (c) After several elementary column operations on $\det(\mathbf{p}^1,\mathbf{p}^2,\mathbf{a})$, the dot product of the stationary vector $\pi$ and an arbitrary vector $\mathbf{a}=(a_1,a_2,\dots,a_9)^T$ is equal to $\det(\mathbf{p}^1,\mathbf{p}^2,\mathbf{a})$, where the first columns $\hat{\mathbf{p}}^1$ is only controlled by organization $1$.}
\label{fig:example} 
\end{figure*}

In particular, when $\mathbf{a}=\mathbf{u}^i$, then organization $i$'s expected utility in the stationary state is:
\begin{equation}\label{eq:u_a}
E^i=\frac{\pi^T\cdot\mathbf{u}^i}{\pi^T\cdot\textbf{1}}=\frac{\det (\mathbf{p}^1,\dots,\mathbf{p}^N,\mathbf{u}^i)}{\det (\mathbf{p}^1,\dots,\mathbf{p}^N,\mathbf{1})},
\end{equation}
which makes a linear combination of all organizations’ expected utilities yielding the following equation:
\begin{equation}\label{eq:zd_equation}
\sum_{x=1}^{N}\alpha_x E^x+\alpha_0=\frac{\det (\mathbf{p}^1,\dots,\mathbf{p}^N,\sum_{x=1}^{N}\alpha_x \mathbf{u}^x+\alpha_0\mathbf{1})}{\det (\mathbf{p}^1,\dots,\mathbf{p}^N,\mathbf{1})}.
\end{equation}
In the above equation, $\alpha_0$ and $\alpha_x, x\in\mathcal{N},$ are constants for the linear combination. Moreover, after some elementary column operations on $\det(\mathbf{p}^1, \mathbf{p}^2,\dots,\mathbf{p}^N,\mathbf{a})$, in which there can be a certain column controlled by a certain organization. Fig. \ref{fig:example}(c) displays the determinant of $\pi^T\cdot\mathbf{a}$ based on the example in Fig. \ref{fig:illustration}, in which the first column is solely determined by organization $1$. We formally introduce the MMZD strategy in the cross-silo FL game below.

\begin{definition}
(MMZD strategy). In the cross-silo FL game, given the utility vectors $\mathbf{u}^x, x=1,\cdots,N$ with $r$ communication rounds, we define the MMZD strategy of organization $i$ as the strategy vector $\hat{\mathbf{p}}^i$ that satisfies the following conditions:
\begin{equation}\label{eq:zd_stratrgy}
\hat{\mathbf{p}}^i=\phi(\sum_{x=1}^{N}\alpha_x \mathbf{u}^x+\alpha_0\mathbf{1}),    
\end{equation}
where $\phi$ is a non-zero constant.
\end{definition}

With $\hat{\mathbf{p}}^i$ being under the control of organization $i$, the corresponding column of $\hat{\mathbf{p}}^i$ and the last column of $\det (\mathbf{p}^1,\dots,\mathbf{p}^N,\sum_{x=1}^{N}\alpha_x \mathbf{u}^x+\alpha_0\mathbf{1})$ will be proportional. And (\ref{eq:zd_equation}) can be converted to:
\begin{equation}\label{eq:zd_zero}
\sum_{x=1}^{N}\alpha_x E^x+\alpha_0=0.
\end{equation}
We further study the social welfare maximization problem with the help of the MMZD in this circumstance. Take organization $1$ performing the MMZD strategy as an example. According to (\ref{eq:zd_zero}), by setting $\alpha_x=1, x\in\mathcal{N}$, the social welfare can be calculated as $\sum_{x=1}^{N} E^x=-\alpha_0$. Thus, the issue of maximizing the social welfare is equivalent to the following optimization problem:
\begin{equation*}\label{eq:min_gamma}
\begin{split}
&\min \alpha_0,\\
&\,s.t.
\begin{cases}
0\le p^1_{j,g}\le 1,j\in\{1,2,...,(r+1)^N\},g\in\{0,...,r\},\\
\hat{\mathbf{p}}^1=\phi(\sum_{x=1}^{N} \mathbf{u}^x+\alpha_0\mathbf{1}),\\
\phi \neq 0.\\
\end{cases}
\end{split}
\end{equation*}

\begin{proposition}
In the cross-silo FL game, organization $1$ could maximize the social welfare by adopting the MMZD strategy $\mathbf{p}^1$. The element in $\mathbf{p}^1$ is calculated by: 
\begin{align*}
p^1_h=
\begin{cases}   
\sum_{x=1}^{N} u^x_h+{\alpha_0}_{min}+1, \\\qquad\qquad\qquad h=1,2,...,(r+1)^{N-1},\\
\sum_{x=1}^{N} u^x_h+{\alpha_0}_{min},\\\qquad\qquad\qquad h=(r+1)^{N-1}+1,...,(r+1)^{N+1},\\
\end{cases}
\end{align*}
where $p^1_h$ denotes the $h$-th element in $\mathbf{p}^1$.
\begin{proof}
We denote $u^x_k,k\in\{1,2,...,(r+1)^{N+1}\}$ as the $k$th element in $\mathbf{u}^x$, then we can solve the above optimization problem by considering the following two cases:\\
\textbf{1) $\phi>0:$}\\ 
To meet the constraint $p^1_{j,g}\ge0$, we can get the lower bound of $\alpha_0$ as follows:
\begin{align*}
& {\alpha_0}_{min}=\max(\Lambda_k),\forall{k}\in\{1,2,...,(r+1)^{N+1}\},\\
& \Lambda_k=
\begin{cases}   
-\sum_{x=1}^{N} u^x_k-\frac{1}{\phi}, k=1,2,...(r+1)^{N-1},\\
-\sum_{x=1}^{N} u^x_k, k=(r+1)^{N-1}+1,...,(r+1)^{N+1}.\\
\end{cases}
\end{align*}
To meet the constraint $p^1_{j,g}\le1$, we can get the upper bound of $\alpha_0$ as follows:
\begin{align*}
& {\alpha_0}_{max}=\min(\Lambda_l),\forall{l}\in\{(r+1)^{N+1}+1,...,2(r+1)^{N+1}\},\\
& \Lambda_l= \Lambda_{k+(r+1)^N}\\
& =
\begin{cases}
-\sum_{x=1}^{N} u^x_k, &k=1,2,...(r+1)^{N-1},\\
-\sum_{x=1}^{N} u^x_k+\frac{1}{\phi}, & k=(r+1)^{N-1}+1,...,(r+1)^{N+1}.\\
\end{cases}
\end{align*}
Only if ${\alpha_0}_{min}\le{\alpha_0}_{max}$, can $\alpha_0$ have a feasible solution, which is equivalent to $\max(\Lambda_k)\le \min(\Lambda_l),\forall{k}\in\{1,2,...,(r+1)^{N+1}\},\forall{l}\in\{(r+1)^N+1,...,2(r+1)^{N+1}\}$. If there exists $\phi>0$ satisfying the above constraint, we can obtain the minimum value of $\alpha_0$ as follow:
\begin{multline}\label{eq:case1_alphamin}
{\alpha_0}_{min}=\max\{-\sum_{x=1}^{N} u^x_1-\frac{1}{\phi},...,-\sum_{x=1}^{N} u^x_{(r+1)^{N-1}}-\frac{1}{\phi},\\
-\sum_{x=1}^{N} u^x_{(r+1)^{N-1}+1},...,-\sum_{x=1}^{N} u^x_{(r+1)^{N+1}}\}.
\end{multline}\\
\textbf{2) $\phi<0:$}\\ 
Similarly, when $p^1_{j,g}\ge0$, we have ${\alpha_0}_{min}=\max(\Lambda_l),\forall{l}\in\{(r+1)^{N+1}+1,...,2(r+1)^{N+1}\}$; considering $p^1_{j,g}\le1$, we have ${\alpha_0}_{max}=\min(\Lambda_k),\forall{k}\in\{1,2,...,(r+1)^{N+1}\}$. In addition, $\alpha_0$ is feasible only when ${\alpha_0}_{min}\le{\alpha_0}_{max}$, i.e., $\max(\Lambda_l)\le \min(\Lambda_k),\forall{k}\in\{1,2,...,(r+1)^{N+1}\},\forall{l}\in\{(r+1)^{N+1}+1,...,2(r+1)^{N+1}\}$. Finally, we can get the following result:
\begin{multline}\label{eq:case2_alphamin}
{\alpha_0}_{min}=\max\{-\sum_{x=1}^{N} u^x_1,...,-\sum_{x=1}^{N} u^x_{(r+1)^{N-1}},\\
-\sum_{x=1}^{N} u^x_{(r+1)^{N-1}+1}+\frac{1}{\phi},...,-\sum_{x=1}^{N} u^x_{(r+1)^{N+1}}+\frac{1}{\phi}\}.
\end{multline}
In summary, by (\ref{eq:case1_alphamin}) and (\ref{eq:case2_alphamin}), organization $1$ can unilaterally set the expected social welfare $\sum_{x=1}^{N} E^x$ with the MMZD strategy $\mathbf{p}^1$ meeting $\hat{\mathbf{p}}^1=\phi(\sum_{x=1}^{N} \mathbf{u}^x+\alpha_0\mathbf{1})$, with each element of $\mathbf{p}^1$ calculated by:
\begin{align*}\label{eq:p_i_calculate2}
p^1_h=
\begin{cases}   
\sum_{x=1}^{N} u^x_h+{\alpha_0}_{min}+1, \\\qquad\qquad\qquad h=1,2,...,(r+1)^{N-1},\\
\sum_{x=1}^{N} u^x_h+{\alpha_0}_{min},\\\qquad\qquad\qquad h=(r+1)^{N-1}+1,...,(r+1)^{N+1},\\
\end{cases}
\end{align*}
where $p^1_h$ denotes the $h$-th element in $\mathbf{p}^1$.
\end{proof}
\end{proposition}

\section{Social welfare maximization by MMZD alliance}
In the previous section, we have proved that by adopting MMZD strategies, individual organizations can maximize social welfare in the cross-silo FL game. 
According to (\ref{eq:zd_stratrgy}), however, one can see that every organization can deploy the MMZD strategy to control the social welfare. Therefore, it is possible that multiple organizations utilize the MMZD strategy at the same time. Then how will this affect the control of the social welfare maximization in the cross-silo FL game? Specifically, will it lead to stronger control and further enlarge the maximum social welfare value? 

In this section, we explore the case of multiple organizations playing the MMZD strategy to form an alliance in the cross-silo FL game. We call them MMZD Alliance (MMZDA) organizations (denoted as $\mathcal{A}$), assuming that all MMZDA organizations use the same MMZD strategy to prevent the low social welfare. Besides, we define the other organizations as outsider organizations (denoted as $\mathcal{N}\backslash\mathcal{A}$) that may inactively participate in communication rounds and just want to get model trained by other organizations for free.

For the sake of convenience, we assume that organization $i\in\mathcal{A},|\mathcal{A}|=N^{\mathcal{A}}$ is an alliance, which performs the same MMZD strategy as the leader organization $a\in\mathcal{A}$. As for an outsider organization $j\in\mathcal{N}\backslash\mathcal{A}, |\mathcal{N}\backslash\mathcal{A}|=N-N^{\mathcal{A}}$, it may not participate in any communication round. 

In this new scenario based on the MMZDA, we pay more attention to strategies and behaviors rather than the organizations themselves. Since the MMZDA members take the same actions, we treat them as an entity represented by the leader organization $a$. In the alliance, organizations maintain strategy consistency through certain information exchange and communication, which is summarized in Algorithm \ref{algo:MMZDA}. At the beginning of each game round, other organizations $j\in\mathcal{A}\backslash \{a\}$ send their utility functions ${U}^j$ to the leader organization $a$ (line \ref{algo:line1}). Then, the leader organization $a$ calculates the utility vector $\mathbf{v}^i$ by these utility functions (line \ref{algo:line2}). Next, the leader organization $a$ solves the optimization problem (\ref{eq:min_gamma_2}) to obtain the corresponding strategy $\hat{\mathbf{q}}^a$ of MMZDA (line \ref{algo:find}), which is then sent to the other organizations in the alliance (line \ref{algo:line4}). Upon receiving the strategy $\hat{\mathbf{q}}^a$, the other organizations adopt it as the strategy for the current game round and implement it (line \ref{algo:line5}).

\begin{algorithm}
\caption{MMZDA}\label{algo:MMZDA}
\LinesNumbered
\raggedright
\SetKwInOut{Input}{Input}
\SetKwInOut{Output}{Output}
\textbf{Input: } total  global communication round $r$, the utility functions ${U}^j(\mathbf{y})$ of all organizations
\\

\nonl \textbf{Other organizations $j\in\mathcal{A}\backslash \{a\}$: }\\
\text{   send the utility function} ${U}^j$ \text{to the leader organization $a$}\\ \label{algo:line1}
\nonl \textbf{Leader organization $a$: }\\
\text{   generate} $\mathbf{v}^a$ \text{from} $\sum_{j\in\mathcal{A}} {U}^j(\mathbf{y})$\\ \label{algo:line2}
\text{   calculate} $\hat{\mathbf{q}}^a$ \text{by solving optimization problem (\ref{eq:min_gamma_2})}\\ \label{algo:find}
\text{   send} $\hat{\mathbf{q}}^a$ \text{to other organizations in }$\mathcal{A}$\\ \label{algo:line4}
\nonl \textbf{Other organizations $j\in\mathcal{A}\backslash \{a\}$: }\\
\text{   implement strategy} $\hat{\mathbf{p}}^j \gets\hat{\mathbf{q}}^a$\\ \label{algo:line5}

\end{algorithm}

In each game round, any outsider organization or alliance leader (organization $a$) can choose the action $y_i\in\{0,1,...,r\}$, so there are $(r+1)^{N-N^{\mathcal{A}}}$ possible outcomes for each game round. We assume that the organizations have one-round-memory. And we define $\mathcal{N}^c=\mathcal{N}\backslash\mathcal{A}\cup\{a\}$ as the players of the cross-silo FL game based on MMZDA. For arbitrary organization $i\in\mathcal{N}^c$, its mixed strategy $\mathbf{q}^i$ is defined as:
\begin{equation}
\mathbf{q}^i=[q^i_{1,0},q^i_{1,1},...,q^i_{1,r},q^i_{2,0},...,q^i_{j,g},...,q^i_{(r+1)^{N-N^{\mathcal{A}}+1},r}]^T,
\end{equation}
where $q^i_{j,g} (j\in\{1,2,...,(r+1)^{N-N^{\mathcal{A}}+1}\},g\in\{0,1,...,r\})$ represents the probability of organization $i$ choosing action $y_i=g$ in the current game round conditioning on the $j$-th outcome of the previous game round. In addition, the corresponding utility vector $\mathbf{v}^i$ is denoted as:
\begin{equation}
\mathbf{v}^i=[v^i_{1,0},v^i_{1,1},...,v^i_{1,r},v^i_{2,0},...,v^i_{j,g},...,v^i_{(r+1)^{N-N^{\mathcal{A}}+1},r}]^T,
\end{equation}
where each utility $v^i_{j,g}$ of organization $i$ choosing action $y_i=g$ in the $j$-th outcome can be calculated by $v^i_{j,g}=U^i(\mathbf{y}^{(j,g)})$, with $\mathbf{y}^{(j,g)}$ denoting the action vector $\mathbf{y}$ corresponding to the $j$-th outcome but $y_i=g$. Moreover, if $i=a$, then $v^a_{j,g}=\sum_{x\in{\mathcal{A}}}U^x(\mathbf{y}^{(j,g)})$. In Section \ref{maxwelfare}, we perform that the linear combination of all organizations' expected utilities can be represented as (\ref{eq:u_a}) and (\ref{eq:zd_equation}).
Similarly, we can draw that
\begin{multline}
\label{eq:zd_equation_a}
\sum_{x\in{\mathcal{N}^c}}\gamma_x F^x+\gamma_0\\
=\frac{\det (\mathbf{q}^1,...,\mathbf{q}^i,...,\mathbf{q}^{N-N^{\mathcal{A}}+1},\sum_{x\in{\mathcal{N}^c}}\gamma_x \mathbf{v}^x+\gamma_0\mathbf{1})}{\det (\mathbf{q}^1,...,\mathbf{q}^i,...,\mathbf{q}^{N-N^{\mathcal{A}}+1},\mathbf{1})},
\end{multline}
where $\gamma_0\in\mathbb{R}$ as well as $\gamma_x\in\mathbb{R},x\in{\mathcal{N}^c}$ are constants, and $F^x$ denotes the expected utility of organization $x$.. Thus, when organization $i$ chooses a strategy that satisfies $\hat{\mathbf{q}}^i=\phi(\sum_{x\in{\mathcal{N}^c}}\gamma_x \mathbf{v}^x+\gamma_0\mathbf{1})$, where $\phi$ is a non-zero constant and $\hat{\mathbf{q}}^i$ is under the control of organization $i$, the column related to $\hat{\mathbf{q}}^i$ and the last column of $\det (\mathbf{q}^1,...,\mathbf{q}^i,...,\mathbf{q}^N,\mathbf{v}^i)$ will be proportional. Then (\ref{eq:zd_equation}) can be converted to:
\begin{equation}\label{eq:zd_zero_a}
\sum_{x\in{\mathcal{N}^c}}\gamma_x F^x+\gamma_0=0.
\end{equation}

In order to investigate the social welfare maximization problem by MMZDA, we rewrite the optimization problem as:
\begin{equation}\label{eq:min_gamma_2}
\begin{split}
&\min \gamma_0,\\
&\,s.t.
\begin{cases}
0\le q^a_{j,g}\le 1,\\
j\in\{1,2,...,(r+1)^{N-N^{\mathcal{A}}+1}\},g\in\{0,...,r\},\\
\hat{\mathbf{q}}^a=\phi(\sum_{x\in{\mathcal{N}^c}} \mathbf{v}^x+\gamma_0\mathbf{1}),\\
\phi \neq 0.\\
\end{cases}
\end{split}
\end{equation}

\begin{proposition}
In the cross-silo FL game, an MMZD alliance $\mathcal{A}$ with the leader organization $a$ could maximize the social welfare by adopting the strategy $\mathbf{q}^a$. The element in $\mathbf{q}^a$ is calculated by:
\begin{align*}
q^a_h=
\begin{cases}   
\sum_{x=1}^{N} v^x_h+{\gamma_0}_{min}+1,\\\qquad\quad h=1,2,...,(r+1)^{N-N^{\mathcal{A}}},\\
\sum_{x=1}^{N} v^x_h+{\gamma_0}_{min},\\\qquad\quad h=(r+1)^{N-N^{\mathcal{A}}}+1,...,(r+1)^{N-N^{\mathcal{A}}+2},\\
\end{cases}
\end{align*}
where $q^a_h$ denotes the $h$-th element in $\mathbf{q}^a$.
\begin{proof}
Similar to Section \ref{maxwelfare}, we denote $v^x_k,k\in\{1,2,...,(r+1)^{N-N^{\mathcal{A}}+2}\}$ as the $k$th element in $\mathbf{v}^x$, then we can solve this optimization problem by discussing these two situations:\\
\textbf{1)} $\phi>0:$\\
When $q^a_{j,g}\ge0$, we have ${\gamma_0}_{min}=\max(\Lambda_l),\forall{l}\in\{(r+1)^{N-N^{\mathcal{A}}+2}+1,...,2(r+1)^{N-N^{\mathcal{A}}+2}\}$; given $q^a_{j,g}\le1$, we have ${\gamma_0}_{max}=\min(\Lambda_k),\forall{k}\in\{1,2,...,(r+1)^{N-N^{\mathcal{A}}+2}\}$. In addition, $\gamma_0$ is feasible only when ${\gamma_0}_{min}\le{\gamma_0}_{max}$, i.e., $\max(\Lambda_l)\le \min(\Lambda_k),\forall{k}\in\{1,2,...,(r+1)^{N-N^{\mathcal{A}}+2}\},\forall{l}\in\{(r+1)^{N-N^{\mathcal{A}}+2}+1,...,2(r+1)^{N-N^{\mathcal{A}}+2}\}$. Finally, we can get the following result:
\begin{multline}\label{eq:case1_gammamin}
{\gamma_0}_{min}=\max\{-\sum_{x\in{\mathcal{N}^c}} v^x_1-\frac{1}{\phi},...,-\sum_{x\in{\mathcal{N}^c}} v^x_{(r+1)^{N-N^{\mathcal{A}}}}-\frac{1}{\phi},\\
-\sum_{x\in{\mathcal{N}^c}} v^x_{(r+1)^{N-N^{\mathcal{A}}}+1},...,-\sum_{x\in{\mathcal{N}^c}} v^x_{(r+1)^{N-N^{\mathcal{A}}+2}}\}.
\end{multline}
\textbf{2)} $\phi<0:$\\
Given $q^a_{j,g}\ge0$, we have ${\gamma_0}_{min}=\max(\Lambda_l),\forall{l}\in\{(r+1)^{N-N^{\mathcal{A}}+2}+1,...,2(r+1)^{N-N^{\mathcal{A}}+2}\}$; while $q^a_{j,g}\le1$, we can get ${\gamma_0}_{max}=\min(\Lambda_k),\forall{k}\in\{1,2,...,(r+1)^{N-N^{\mathcal{A}}+2}\}$. In addition, $\gamma_0$ is feasible only when ${\gamma_0}_{min}\le{\gamma_0}_{max}$, i.e., $\max(\Lambda_l)\le \min(\Lambda_k),\forall{k}\in\{1,2,...,(r+1)^{N-N^{\mathcal{A}}+2}\},\forall{l}\in\{(r+1)^{N-N^{\mathcal{A}}+2}+1,...,2(r+1)^{N-N^{\mathcal{A}}+2}\}$. Finally, we can get the following result:
\begin{multline}\label{eq:case2_gammamin}
{\gamma_0}_{min}=\max\{-\sum_{x\in{\mathcal{N}^c}} v^x_1,...,-\sum_{x\in{\mathcal{N}^c}} v^x_{(r+1)^{N-N^{\mathcal{A}}}},\\
-\sum_{x\in{\mathcal{N}^c}} v^x_{(r+1)^{N-N^{\mathcal{A}}}+1}+\frac{1}{\phi},...,-\sum_{x\in{\mathcal{N}^c}} v^x_{(r+1)^{N-N^{\mathcal{A}}+2}}+\frac{1}{\phi}\}.
\end{multline}
In summary, by (\ref{eq:case1_gammamin}) and (\ref{eq:case2_gammamin}), the alliance can unilaterally set the expected social welfare $\sum_{x\in{\mathcal{N}^c}} E^x$ with the MMZD strategy $\mathbf{q}^a$ meeting $\hat{\mathbf{q}}^a=\phi(\sum_{x\in{\mathcal{N}^c}} \mathbf{v}^x+\gamma_0\mathbf{1})$, with each element of $\mathbf{q}^a$ calculated by:
\begin{align*}\label{eq:q_i_calculate}
q^a_h=
\begin{cases}   
\sum_{x=1}^{N} v^x_h+{\gamma_0}_{min}+1,\\\qquad\quad h=1,2,...,(r+1)^{N-N^{\mathcal{A}}},\\
\sum_{x=1}^{N} v^x_h+{\gamma_0}_{min},\\\qquad\quad h=(r+1)^{N-N^{\mathcal{A}}}+1,...,(r+1)^{N-N^{\mathcal{A}}+2},\\
\end{cases}
\end{align*}
where $q^a_h$ denotes the $h$-th element in $\mathbf{q}^a$.
\end{proof}
\end{proposition}
Moreover, we further discover that the maximum social welfare under MMZDA is larger than that controlled by single MMZD organization.

\begin{theorem}\label{th:MMZDalliance}
In the cross-silo FL game, the social welfare can achieve larger maximum value by MMZDA than that by single MMZD organization.
\end{theorem}
\begin{proof}
In the cross-silo FL game, we can achieve a maximum social welfare $\sum_{x=1}^{N} E^x=-\alpha_0$ using MMZD strategy by an individual organization. Further, we are able to draw another maximum social welfare $\sum_{x\in{\mathcal{N}^c}}\gamma_x F^x=-\gamma_0$ by MMZDA with studying the following two cases:\\
\textbf{1) $\phi>0:$}\\ 
In this case, we have:
\begin{align*}
&{\alpha_0}_{min}=\max\{-\sum_{x=1}^{N} u^x_1-\frac{1}{\phi},...,-\sum_{x=1}^{N} u^x_{(r+1)^{N-1}}-\frac{1}{\phi},\\
&-\sum_{x=1}^{N} u^x_{(r+1)^{N-1}+1},...,-\sum_{x=1}^{N} u^x_{(r+1)^{N+1}}\}.\\
&{\gamma_0}_{min}=\max\{-\sum_{x\in{\mathcal{N}^c}} v^x_1-\frac{1}{\phi},...,-\sum_{x\in{\mathcal{N}^c}} v^x_{(r+1)^{N-N^{\mathcal{A}}}}-\frac{1}{\phi},\\
&-\sum_{x\in{\mathcal{N}^c}} v^x_{(r+1)^{N-N^{\mathcal{A}}}+1},...,-\sum_{x\in{\mathcal{N}^c}} v^x_{(r+1)^{N-N^{\mathcal{A}}+2}}\}.
\end{align*}
The value of ${\alpha_0}_{min}$ is the maximum value in the certain $(r+1)^{N+1}$ values, and we denote these candidate values as set $X_1$. While the value of ${\gamma_0}_{min}$ generating from the $(r+1)^{N-N^{\mathcal{A}}+2}$ values, we denote this candidate values as set $X_2$. Note that the $(r+1)^{N+1}$ values cover all possible outcomes, but the $(r+1)^{N-N^{\mathcal{A}}+2}$ values do not cover all outcomes, since the MMZDA organizations employ the same strategy. Given a certain element like $-\sum_{x\in{\mathcal{N}^c}} v^x_{k}-\frac{1}{\phi} \in X_2, k\in \{1,2,...,N-N^{\mathcal{A}}\}$, as $v^a_{j,g}=\sum_{x\in{\mathcal{A}}}U^x(\mathbf{y}^{(j,g)})$, we have:
\begin{multline}
-\sum_{x\in{\mathcal{N}^c}} v^x_{k}-\frac{1}{\phi}=\sum_{x\in{\mathcal{A}}}U^x(\mathbf{y}^{(j,g)})+-\sum_{x\in{\mathcal{N}^c}} v^x_{k}-\frac{1}{\phi}\\
=\sum^{N}_{x=1}U^x(\mathbf{y}^{(j\times(r+1)^{N^{\mathcal{A}-1}}-1,g)})-\frac{1}{\phi}
=-\sum_{x=1}^{N} u^x_{k'}-\frac{1}{\phi},
\end{multline}
where $k'=(r+1)^{N^{\mathcal{A}+k-1}}$.
Otherwise, given an element like $-\sum_{x\in{\mathcal{N}^c}} v^x_{k} \in X_2, k\in \{N-N^{\mathcal{A}}+1,...,N-N^{\mathcal{A}}+2\}$, we can draw:
\begin{align*}
&-\sum_{x\in{\mathcal{N}^c}} v^x_{k}=\sum_{x\in{\mathcal{A}}}U^x(\mathbf{y}^{(j,g)})+-\sum_{x\in{\mathcal{N}^c}} v^x_{k}\\
&=\sum^{N}_{x=1}U^x(\mathbf{y}^{(j-(r+1)^{N-N^{\mathcal{A}-1}})\times(r+1)^{N^{\mathcal{A}-1}}+(r+1)^{N-1},g)})
\\&=-\sum_{x=1}^{N} u^x_{k'},
(k'=(r+1)^{N^{\mathcal{A}+(k-N+N^{\mathcal{A}})-1}}+(r+1)^{N-1})
\end{align*}
As deduced above, the elements in $X_2$ are all in $X_1$, so $X_2$ is a subset of $X_1$. Thus, ${\alpha_0}_{min}\geq {\gamma_0}_{min} $ holds true.\\
\textbf{2) $\phi<0:$}\\ 
When $\phi<0$, the situation is quite similar. 
\begin{align*}
&{\alpha_0}_{min}=\max\{-\sum_{x=1}^{N} u^x_1,...,-\sum_{x=1}^{N} u^x_{(r+1)^{N-1}},\\
&-\sum_{x=1}^{N} u^x_{(r+1)^{N-1}+1}+\frac{1}{\phi},...,-\sum_{x=1}^{N} u^x_{(r+1)^{N+1}}+\frac{1}{\phi}\}.\\
&{\gamma_0}_{min}=\max\{-\sum_{x\in{\mathcal{N}^c}} v^x_1,...,-\sum_{x\in{\mathcal{N}^c}} v^x_{(r+1)^{N-N^{\mathcal{A}}}},\\
&-\sum_{x\in{\mathcal{N}^c}} v^x_{(r+1)^{N-N^{\mathcal{A}}}+1}+\frac{1}{\phi},...,-\sum_{x\in{\mathcal{N}^c}} v^x_{(r+1)^{N-N^{\mathcal{A}}+2}}+\frac{1}{\phi}\}.
\end{align*}
We denote the $(r+1)^{N+1}$ candidate values of ${\alpha_0}_{min}$ as $X_3$. While the value of ${\gamma_0}_{min}$ generating from the $(r+1)^{N-N^{\mathcal{A}}+2}$ values, we denote this candidate values as $X_4$. We can prove that each value of $X_4$ can be found in $X_3$ as above. So ${\alpha_0}_{min}\geq {\gamma_0}_{min} $ holds true as well. Then, we can draw a conclusion that ${\alpha_0}_{min}\geq {\gamma_0}_{min} $, which is equivalent to ${-\alpha_0}_{max}\leq {-\gamma_0}_{max}$. 

Note that ${-\alpha_0}_{max}$ denotes the maximum social welfare controlled by a single organization using the MMZD strategy, and ${-\gamma_0}_{max}$ represents the maximum social welfare achieved by MMZDA. So the social welfare can achieve larger maximization by MMZDA than single MMZD organization, and the proof is concluded.
\end{proof}
The above analysis proves that MMZDA further enhances the ability of non-selfish organizations to maximize the social welfare. If more organizations join the MMZDA, the upper bound of the controllable social welfare increases. 

\section{Experimental Evaluation}
\subsection{Experiment Settings}
\subsubsection{Environment}
In this section, we present the experimental results of our study for the MMZD individual (MMZD) strategy and MMZD Alliance (MMZDA) stategy in social welfare maximization. Generally, all experiments are implemented using Matlab R2021a on a laptop with 2.3 GHz Intel Core i5-8300H processor. In all experiments except for otherwise specification, we set $K=200$, $r=33$, and initial $\phi=0.01$. We adopt the Convolutional Neural Network model on the MNIST, MNIST-Fashion and CIFAR-10 datasets as the global model in cross-silo FL for conducting experiments, among which MNIST is our default dataset. Parameters $\theta_0=23271.584$ and $\theta_1=50193.243$ are derived based on the simulation dataset \cite{li2019convergence}. For every control group with different strategy settings, we repeat the above experiments 100 times, and take the average value as the final expected social welfare. 
\subsubsection{Datasets}
\begin{itemize}
    \item \textbf{MNIST} contains 60,000 data samples, while the testing set has 10,000 data samples. These samples are black and white images of handwritten digits from 0 to 9, with a size of 28$\times$28 pixels.
    \item \textbf{MNIST-Fashion} is a dataset that consists of images depicting various clothing and accessories like shirts, pants, and shoes. It comprises 60,000 training images and 10,000 test images, where each image is associated with one of the 10 different clothing categories. The images in this dataset are also 28$\times$28 black and white.
    \item \textbf{CIFAR-10} has 60,000 color images in total, each having a size of 32$\times$32 pixels. The images are categorized into 10 classes, including airplane, automobile, bird, cat, deer, dog, frog, horse, ship, and truck. Each class contains 6,000 images, resulting in 50,000 training images and 10,000 test images.
\end{itemize}

\subsection{Evaluation of the MMZD Strategy}
First, we evaluate the performance of the MMZD strategy used by individual organization to maximize the social welfare based on simulation experiments. We set $N=10$ since the number of organizations in cross-silo FL is usually small.
\begin{figure}[h]
\centering
  \centering
  \subfigure[MMZD]{
    \label{fig:subfig:a}
    \includegraphics[scale=0.281]{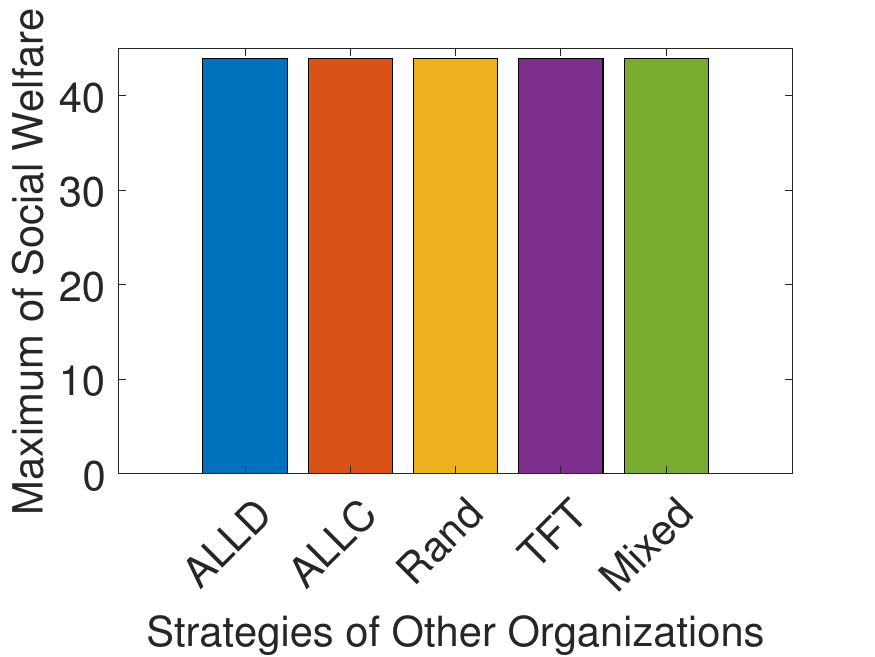}}
  \subfigure[ALLD]{
    \label{fig:subfig:b}
    \includegraphics[scale=0.281]{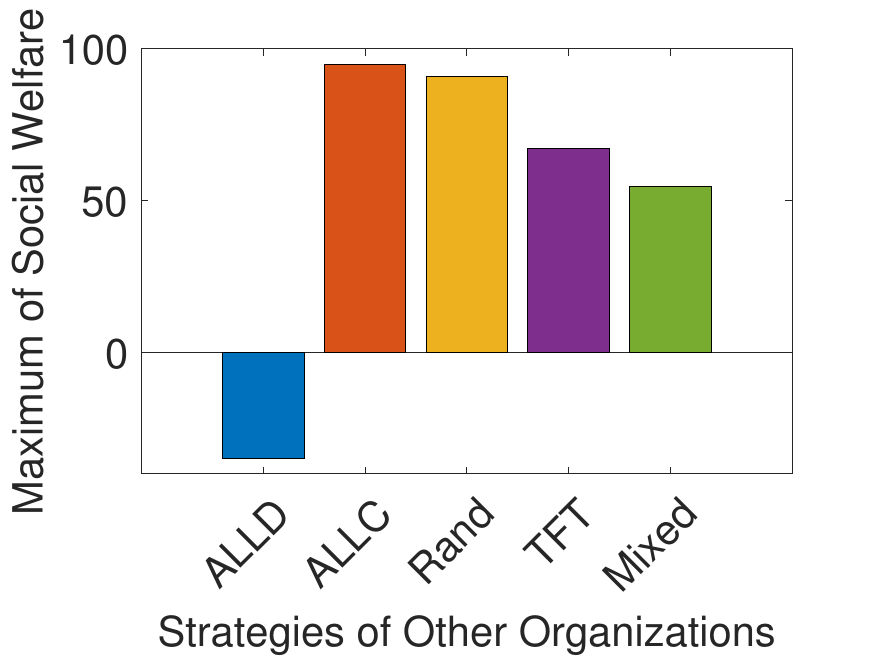}}
  \subfigure[ALLC]{
    \label{fig:subfig:c}
    \includegraphics[scale=0.281]{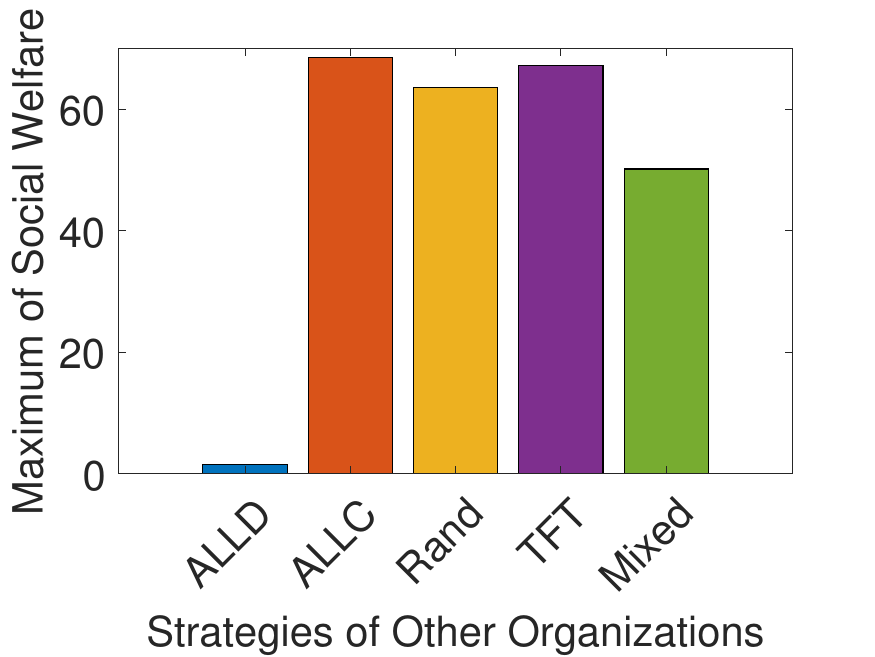}}
  \subfigure[Rand]{
    \label{fig:subfig:d}
    \includegraphics[scale=0.281]{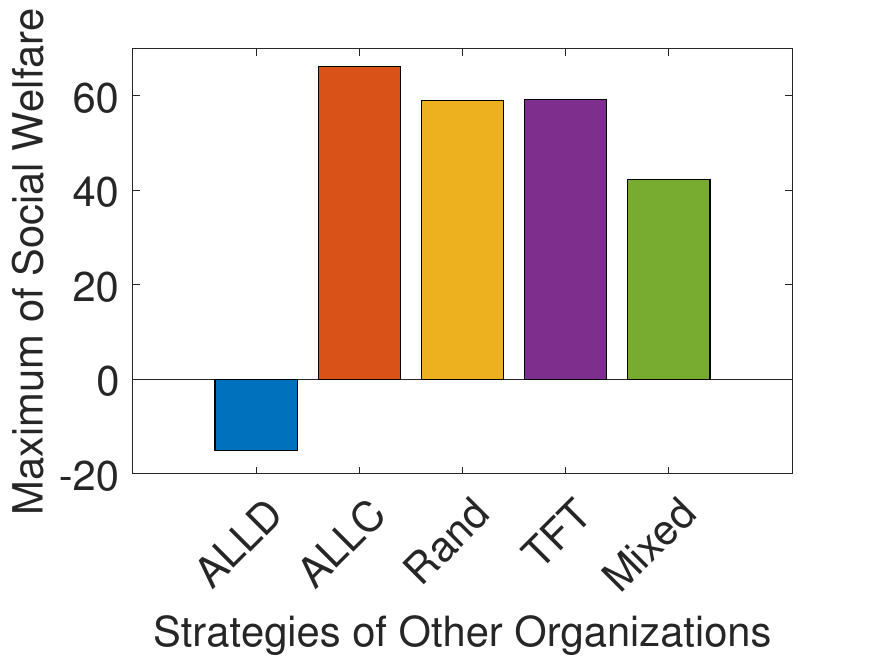}}
\caption{The maximum of the social welfare under different strategy combinations of organization $1$ and others on MNIST.}
  \label{fig:strategies} 
\end{figure}

\begin{figure}[h]
\centering
  \centering
  \subfigure[MMZD]{
    \label{fig:subfig:f-a}
    \includegraphics[scale=0.24]{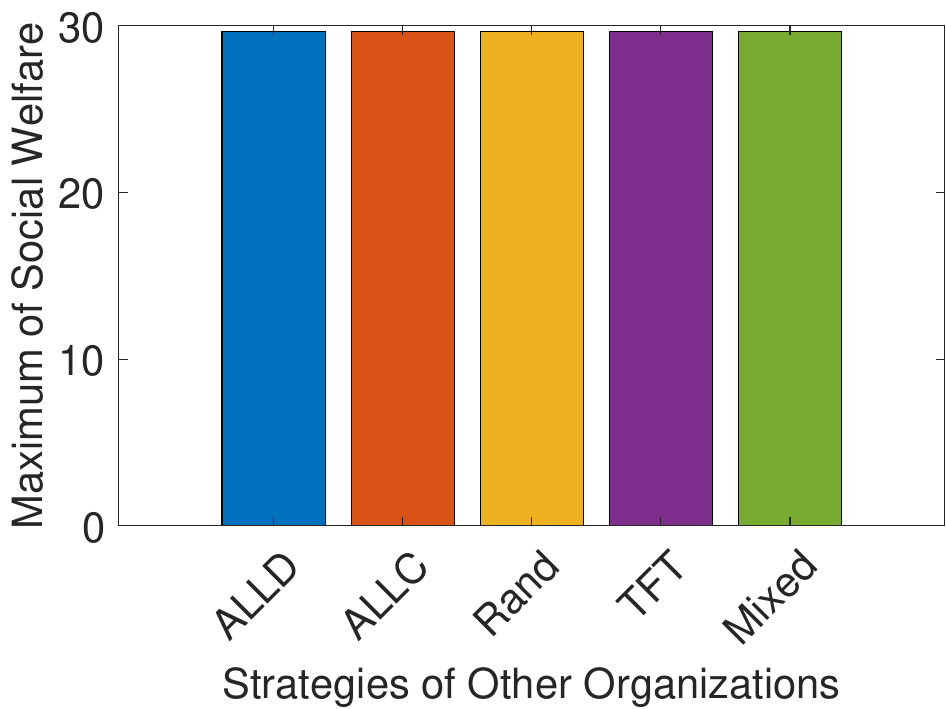}}
  \subfigure[ALLD]{
    \label{fig:subfig:f-b}
    \includegraphics[scale=0.24]{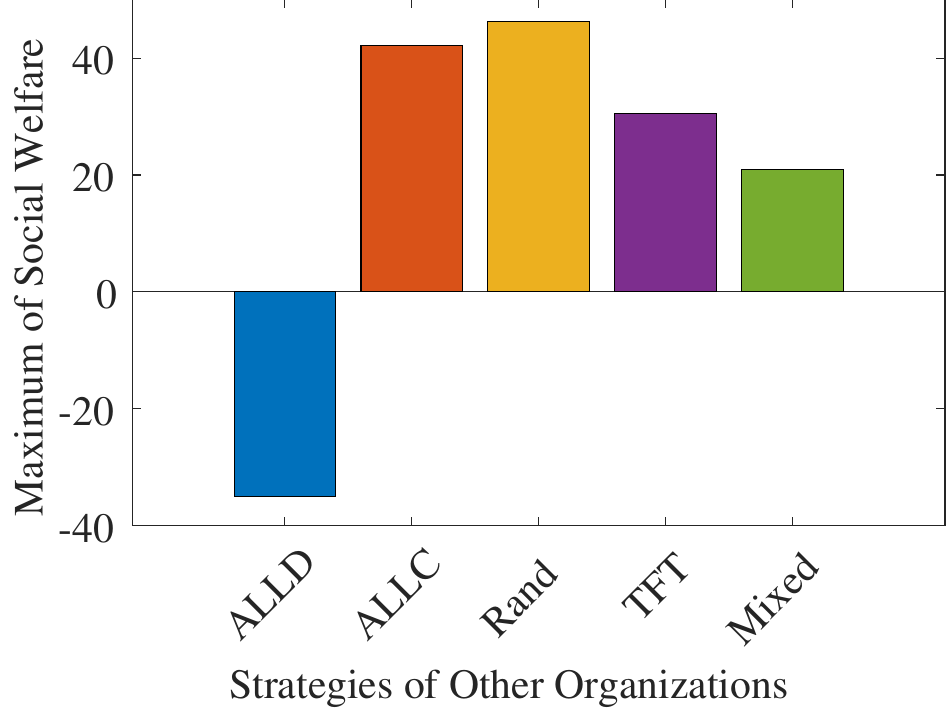}}
  \subfigure[ALLC]{
    \label{fig:subfig:f-c}
    \includegraphics[scale=0.24]{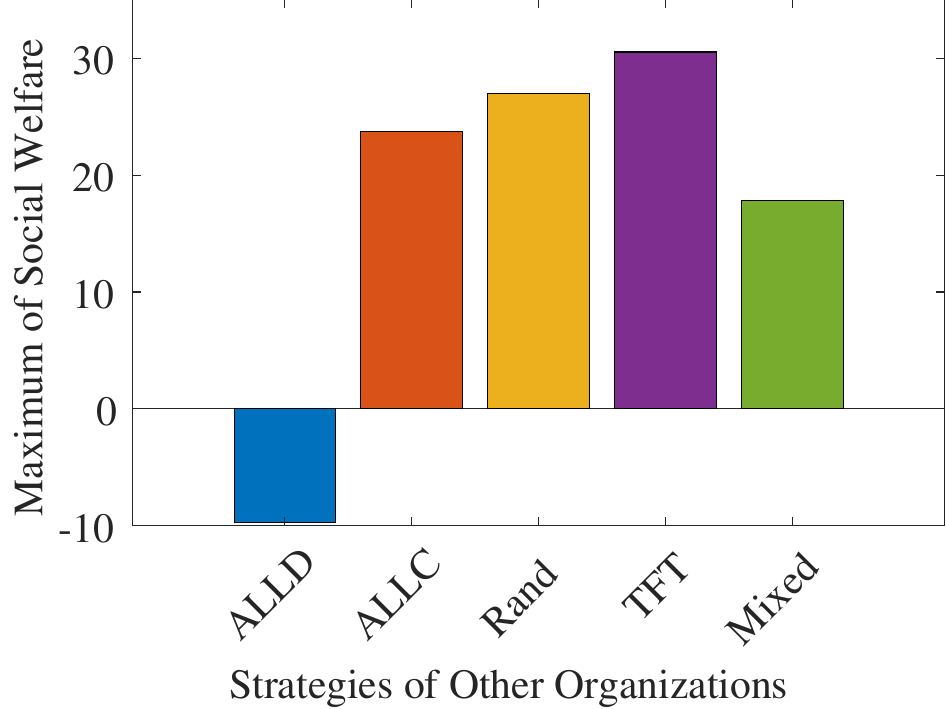}}
  \subfigure[Rand]{
    \label{fig:subfig:f-d}
    \includegraphics[scale=0.24]{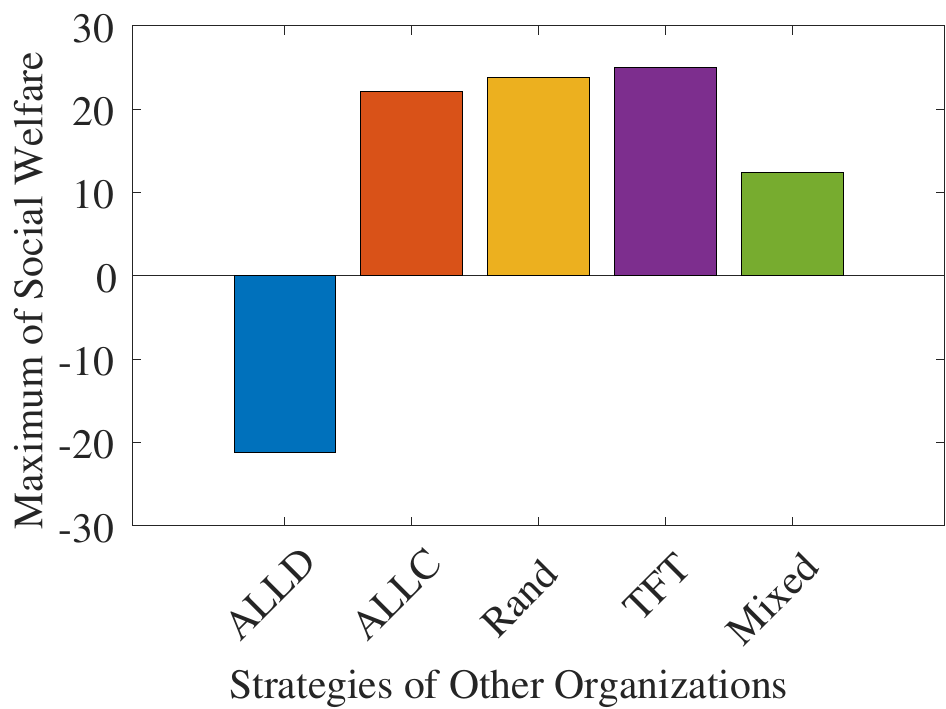}}
\caption{The maximum of the social welfare under different strategy combinations of organization $1$ and others on MNIST-Fashion.}
  \label{fig:MNIST-F} 
\end{figure}

\begin{figure}[h]
\centering
  \centering
  \subfigure[MMZD]{
    \label{fig:subfig:c-a}
    \includegraphics[scale=0.24]{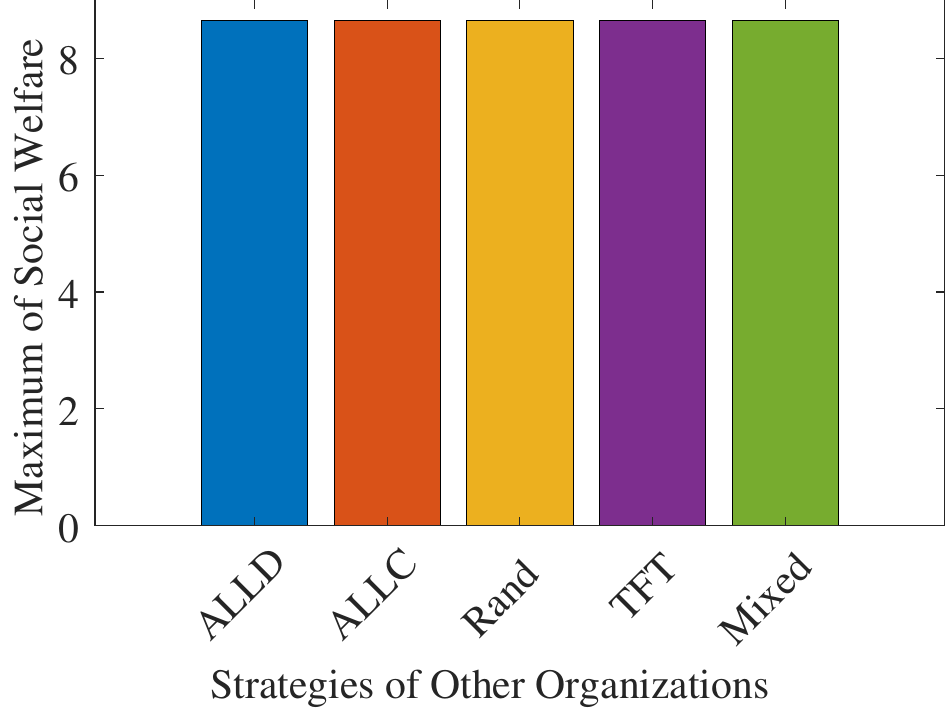}}
  \subfigure[ALLD]{
    \label{fig:subfig:c-b}
    \includegraphics[scale=0.24]{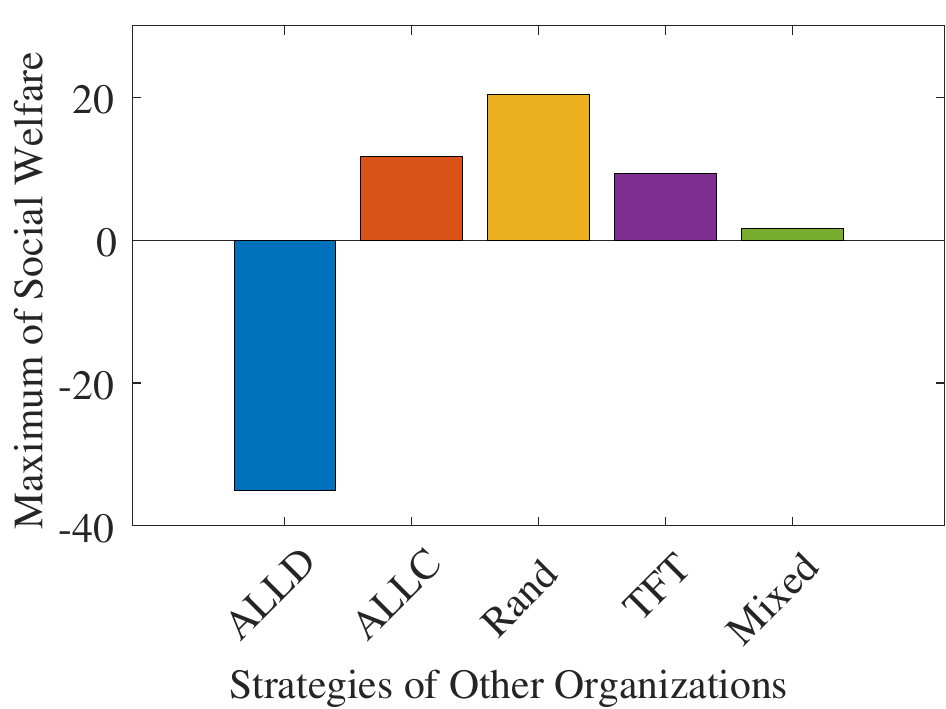}}
  \subfigure[ALLC]{
    \label{fig:subfig:c-c}
    \includegraphics[scale=0.24]{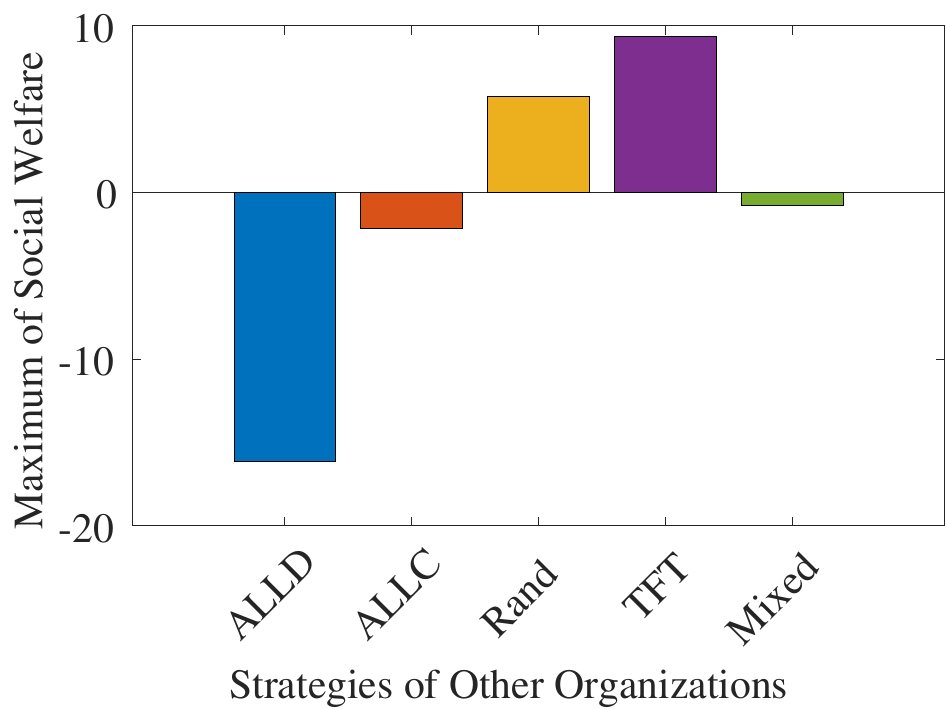}}
  \subfigure[Rand]{
    \label{fig:subfig:c-d}
    \includegraphics[scale=0.24]{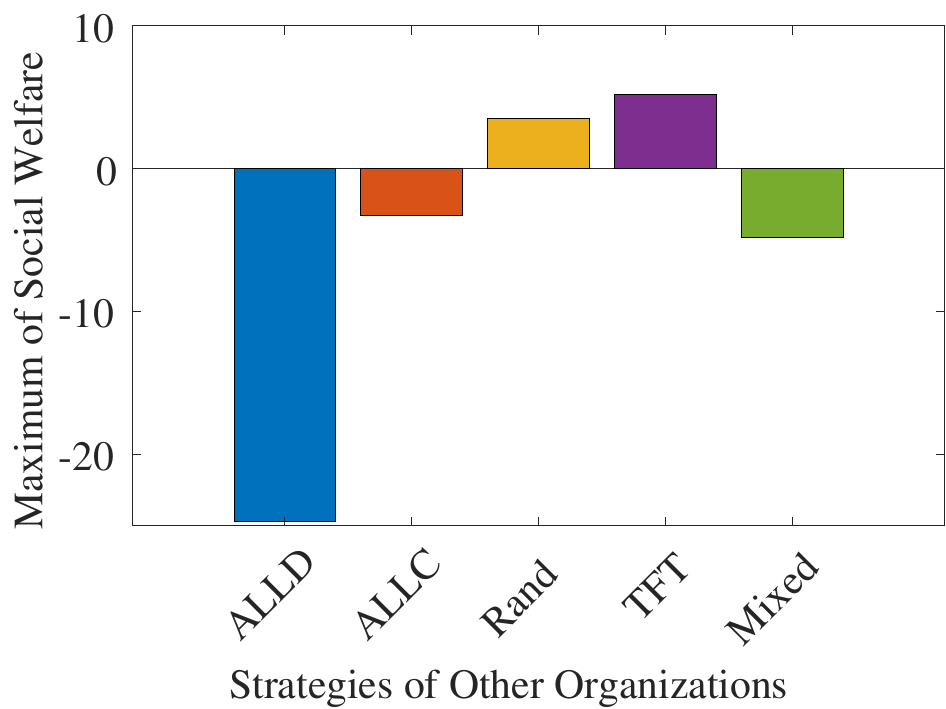}}
\caption{The maximum of the social welfare under different strategy combinations of organization $1$ and others on CIFAR-10.}
  \label{fig:CIFAR} 
\end{figure}

In order to verify the effectiveness of the MMZD strategy on maximizing the social welfare, we compare it with other five classical strategies, by simulating the entire cross-silo FL process for 20 rounds of game. Figs. \ref{fig:strategies}(a)-(d) display the maximum of the social welfare under different strategy combinations of organization $1$ and other organizations on the MNIST dataset. In Fig. \ref{fig:strategies}, organization $1$ adopts MMZD, all-defection (ALLD) \cite{hu2019quality}, all-cooperation (ALLC) \cite{hu2019quality}, and random (Rand) \cite{hu2019quality} strategies. Other organizations use ALLD, ALLC, Rand, Tit-For-Tat (TFT) \cite{nowak1993strategy}, and mixed (Mixed) strategies. Specifically, ALLD strategy is defined as: the organization does not perform local training at all. While ALLC strategy means that the organization participates in all $r$ global aggregation with their local updates in every game round. Organizations which adopt Rand strategy randomly participate in global aggregation from $0$ to $r$ communication rounds with the probability of $\frac{1}{r+1}$. TFT strategy is defined as the organizations randomly choose the number of participating global aggregation from $0$ to $\frac{\lfloor r \rfloor}{2}$ when the sum of communication rounds in last game round is less than $\frac{Nr}{2}$, otherwise they randomly choose the number of participating global aggregation from $\frac{\lfloor r+1 \rfloor}{2}$ to $r$. We define the mixed strategy as adopting a specific strategy chosen from ALLC, ALLD, Rand, and TFT. 

Fig. \ref{fig:strategies}(a) displays the maximum of the social welfare when organization $1$ uses MMZD. We can find that the maximums of the social welfare are the same. The initial action vector is randomly given, but the MMZD strategy can control the social welfare to reach the maximum regardless of others' strategies, verifying the effectiveness of our MMZD scheme. In Fig. \ref{fig:strategies}(b), organization $1$ employs ALLD, which means it acts like a free rider. The social welfare can be large if other organizations participate in global aggregation normally. But if other organizations use ALLD too, the social welfare could be extremely low. This shows the serious impact on the social welfare when all organizations become free riders. In Fig. \ref{fig:strategies}(c), organization $1$ employs ALLC, participating in all global aggregations. However, if other organizations become free riders (adopt ALLD), the effort of organization $1$ cannot turn the tide, and the social welfare is low. Fig. \ref{fig:strategies}(d) displays the maximum of the social welfare when organization $1$ adopts Rand, and the result is similar to Fig. \ref{fig:strategies}(b). 
According to Figs. \ref{fig:strategies}(a)-(d), only the MMZD strategy can stably maximize social welfare regardless of the strategies adopted by other organizations, which demonstrate the strong control ability of the MMZD strategy.

Figs. \ref{fig:MNIST-F}(a)-(d) and Figs. \ref{fig:CIFAR}(a)-(d) show the maximum social welfare controlled by the MMZD strategy compared to that of other strategies on the MNIST dataset and the CIFAR dataset, respectively. Due to the model performance, the social welfare in all cases is reduced, but the MMZD strategy still controls the maximum of the social welfare regardless of the strategies adopted by other organizations. The above experimental results also corroborate the theoretical derivation in Section \ref{maxwelfare} that individual organizations can adopt MMZD strategies to maximize social welfare.

\begin{figure}
\centering
\includegraphics[scale=0.5]{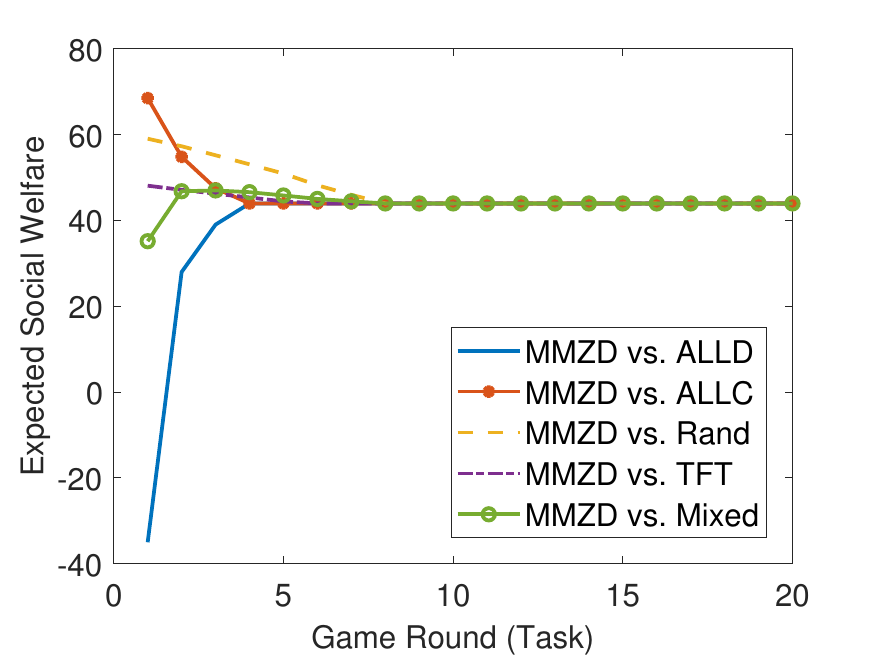}
\caption{Evolution of the expected social welfare.}
\label{fig:ZDconverge} 
\end{figure}
Fig. \ref{fig:ZDconverge} plots the expected social welfare changes in each game round, as organization $1$ adopts the MMZD strategy and other organizations employ different strategies. It is worth noting that the MMZD strategy works quickly as the expected social welfare converges within seven rounds. Fig. \ref{fig:strategies}(a) displays the final result in Fig. \ref{fig:ZDconverge}, which indicates that no matter what kind of strategies other organizations adopt, the social welfare finally converges to a fixed value, verifying the power of the proposed social welfare maximization game.

\subsection{Evaluation of the MMZDA Strategy}

In this subsection, we present the performance of the MMZDA to maximize the social welfare through a series of simulation experiments. And we compare it with the result of a single organization’s MMZD strategy. Meanwhile, we also consider the relative maximum of the social welfare in order to further analyze the control ability of MMZDA.

In Fig. \ref{fig:MMZDalliance vs indi}(a), we set $N^{\mathcal{A}}=4$ and randomly choose four organizations to form an MMZDA, with other parameters unchanged in order to compare with the previous experiment (Fig. \ref{fig:strategies}). From this figure, it's clear that no matter what strategies other organizations adopt, the MMZDA strategy expands the maximum value of the social welfare, comparing with the MMZD strategy performed by a single organization. This experimental result verifies Theorem 
\ref{th:MMZDalliance}. That is, the social welfare can achieve larger maximum value by MMZDA than that by the single MMZD organization. 

Based on the same initial settings, Figs. \ref{fig:MMZDalliance vs indi}(b)-(f) display the evolution process of of the expected social welfare, which is two-fold. The red line represents that the MMZDA strategy is used, while the blue line shows the impact of single organization using the MMZD strategy on the maximum of the social welfare. By comparison, we can find that as the number of game rounds increases, no matter what strategies other organizations adopt, the MMZDA strategy always enables the maximum social welfare gradually converge to a larger fixed value, but it does not have a faster convergence speed. It is worth noting that in Fig. \ref{fig:MMZDalliance vs indi}(e), the two curves do not have the same trend, because in the first game round, the expected social welfare is exactly between the convergent values characterized by MMZD and MMZDA, respectively.

\begin{figure}[h]
\centering
  \centering
  \subfigure[The maximum of the social welfare under MMZDA and MMZD]{
    \label{fig:subfig-2:a}
    \includegraphics[scale=0.26]{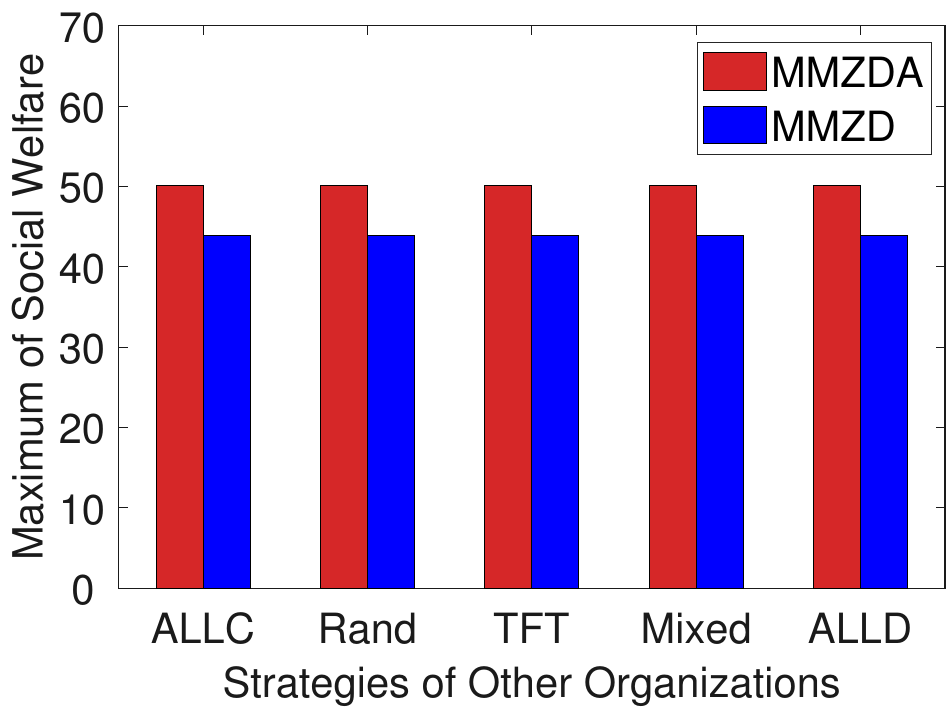}}
  \subfigure[ALLD]{
    \label{fig:subfig-2:b}
    \includegraphics[scale=0.281]{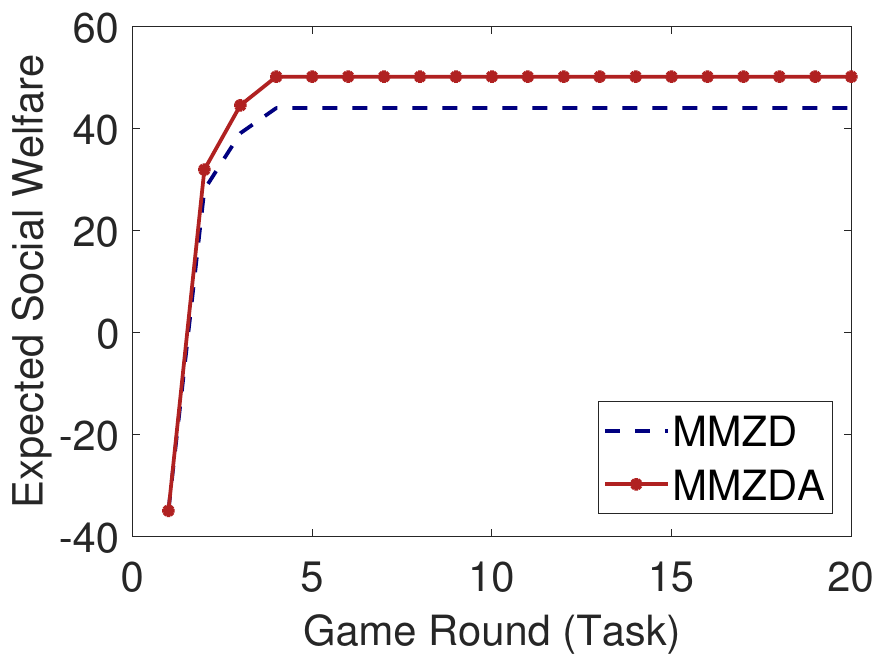}}
  \subfigure[ALLC]{
    \label{fig:subfig-2:c}
    \includegraphics[scale=0.281]{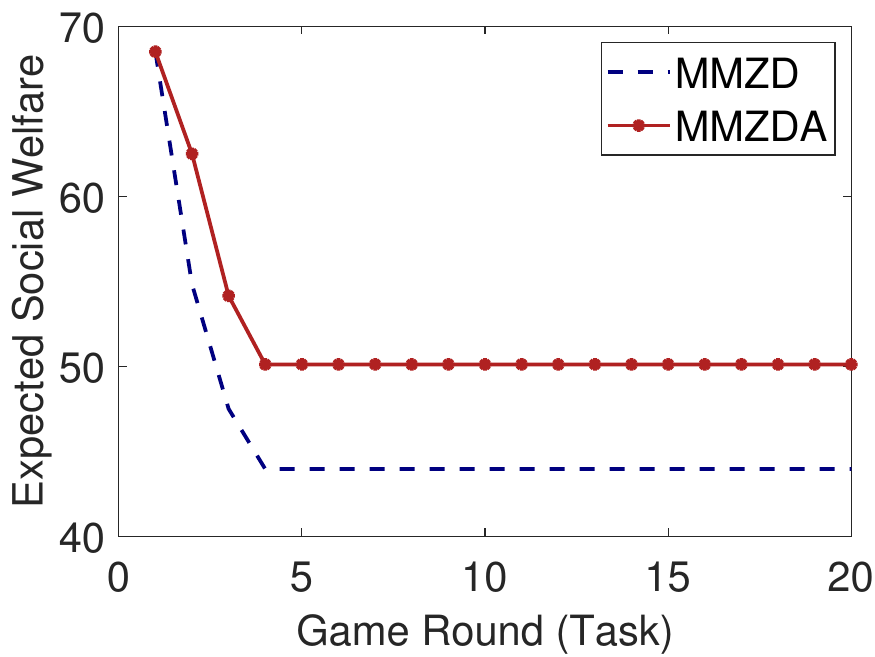}}
  \subfigure[Rand]{
    \label{fig:subfig-2:d}
    \includegraphics[scale=0.281]{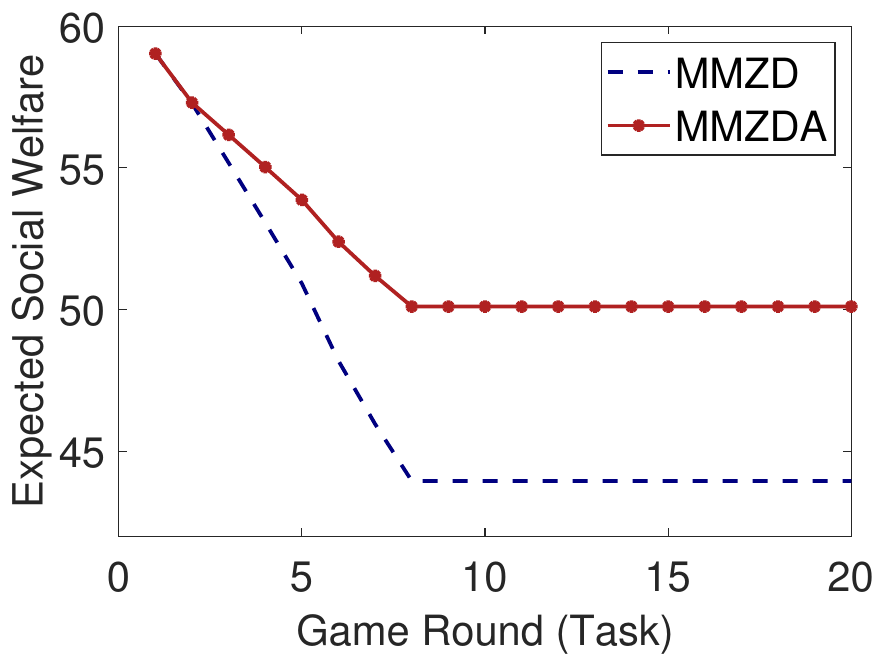}}
  \subfigure[TFT]{
    \label{fig:subfig-2:e}
    \includegraphics[scale=0.281]{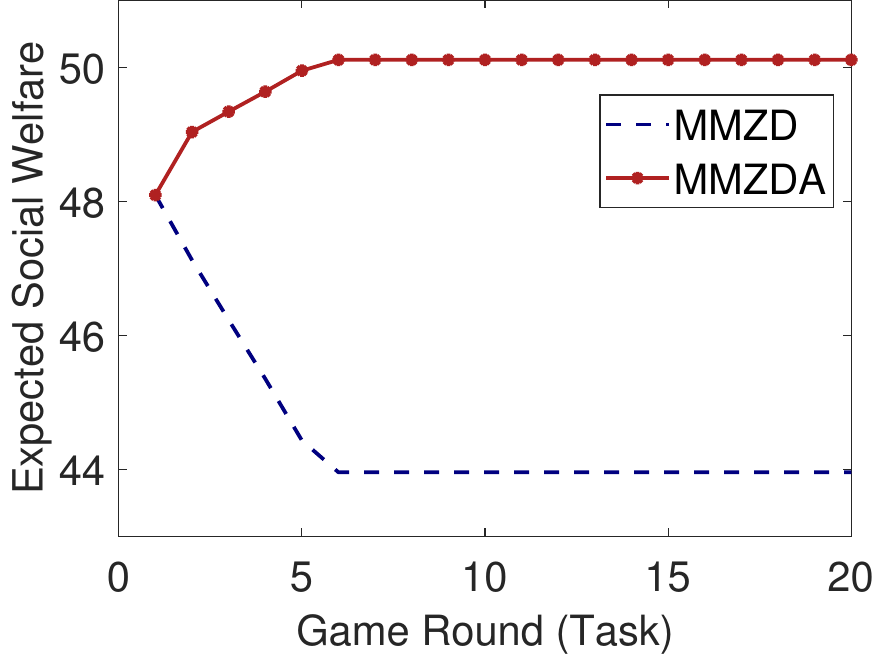}}
  \subfigure[Mixed]{
    \label{fig:subfig-2:f}
    \includegraphics[scale=0.281]{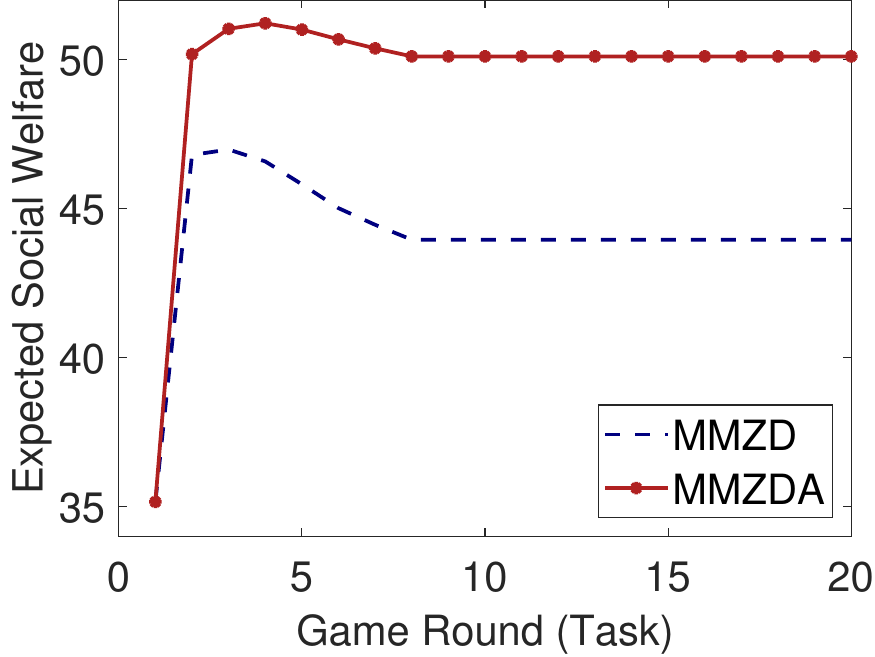}}
\caption{Evolution of the expected social welfare under different strategy combinations, under comparison between the MMZDA strategy and the MMZD individual strategy.}
  \label{fig:MMZDalliance vs indi} 
\end{figure}

In Fig. \ref{fig:changeNAkeepN}, we set $N=10$. Then in Fig. \ref{fig:changeNAkeepN}(a), we explore the changes of the expected social welfare as the number of organizations in the alliance increases when the total number of organizations remains unchanged. Whenever $N^{\mathcal{A}}$ takes a different value, we randomly generate the MMZDA organizations from these 10 organizations. In the histogram, we can conclude that when the total number of organizations does not change, the more organizations that join the MMZD alliance, the higher the maximum social welfare that can be controlled. This also confirms our analysis of the MMZDA strategy, because the increase of $N^{\mathcal{A}}$ expands the range of candidate values in (\ref{eq:case1_gammamin}) and (\ref{eq:case2_gammamin}), thereby increasing the maximum value of the social welfare. Besides, in Fig. \ref{fig:changeNAkeepN}(b), we take the social welfare of all organizations participating in all communication rounds as the absolute maximum value of the social welfare, and study the ratio of the social welfare controlled by MMZDA to the former value. We call this ratio relative maximum of the social welfare. In this case, $N^{\mathcal{A}}$ and the relative maximum of the social welfare are also positively correlated. Together with  Fig. \ref{fig:changeNAkeepN}(a), they show that when $N$ is constant, MMZDA's ability to control the maximum value of the social welfare increases as the number of MMZDA organizations increases. 

\begin{figure}
\centering
  \subfigure[The absolute maximum of the social welfare]{
    \label{fig:changeNAkeepN:abs}
    \includegraphics[scale=0.252]{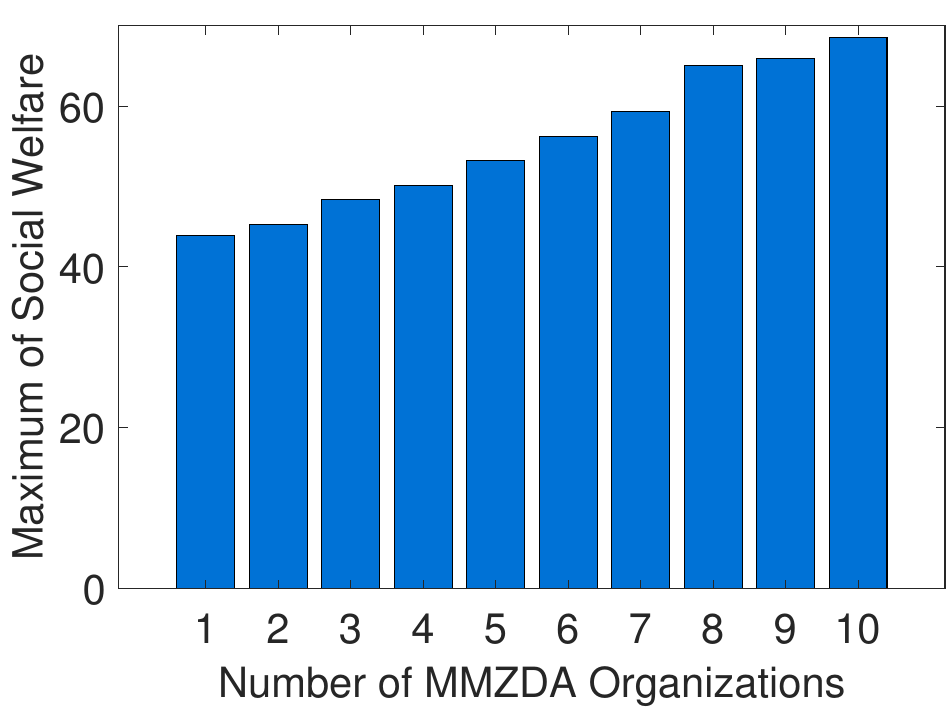}}
  \subfigure[The relative maximum of the social welfare]{
    \label{fig:changeNAkeepN:rel}
    \includegraphics[scale=0.265]{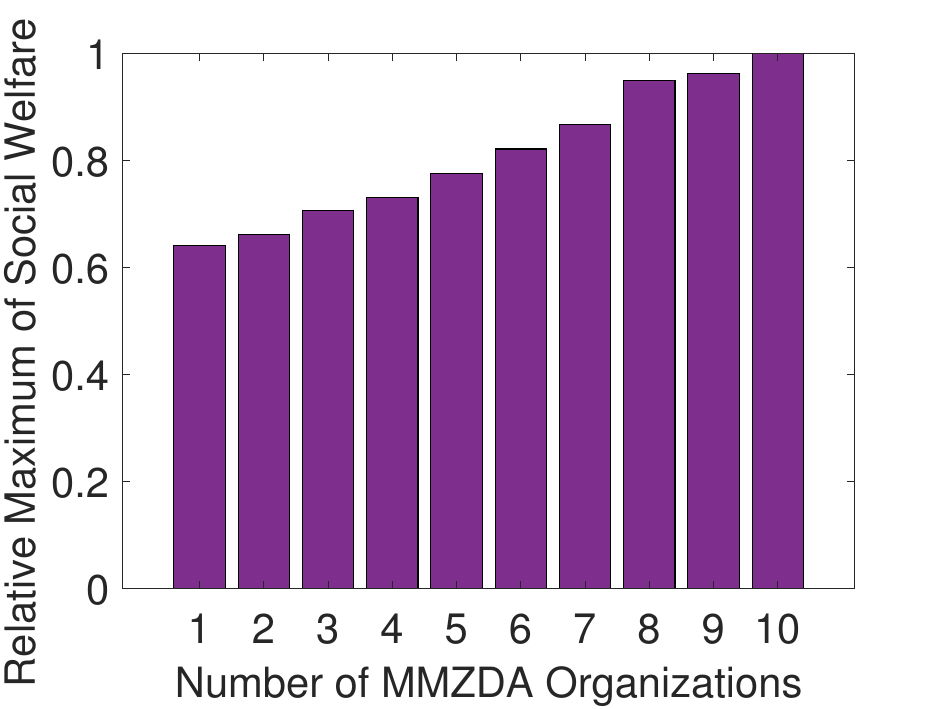}}
\caption{The absolute maximum of the social welfare and the relative maximum of the social welfare as the number of organizations in the MMZDA increases.}
\label{fig:changeNAkeepN} 
\end{figure}

Correspondingly, we continue to investigate the impact of $N$ while $N^{\mathcal{A}}$ does not change. In fact, the change of the total number of organizations $N$ brings a series of differences, including adding new organizations' parameters, changing the organizations’ local datasets, and then changing the coefficients $\theta_0$ and $\theta_1$, which also changes the utility function. For verification, we generate a new simulation dataset, and randomly select $N$ organizations in different situations. We continue to randomly select $N^{\mathcal{A}}=5$ MMZDA organizations from $N$ organizations. According to the new dataset, the corresponding coefficients $\theta_0$ and $\theta_1$ are generated to calculate the final social welfare value. For each distinct $N$, we repeat the process of randomly selecting $N$ organizations 10 times. For each group of selected organizations, we repeat the process of randomly selecting five MMZDA organizations 10 times. Finally, we take the average value as the expected social welfare. As shown in Fig. \ref{fig:changeNkeepNA}(a), we found that simply changing $N$ does not intuitively change the maximum value of the social welfare, because the impact of newly joined organizations on the social welfare is mutative. More specifically, when $N$ is small or even close to $N^{\mathcal{A}}$ (i.e., the case of $N=5$), the maximum value of the social welfare is limited by the total number of organizations. When $N$ is large (i.e., the case of $N=25$), the small MMZD alliance reduces the ability to control the maximum of the social welfare, and cannot bring a large increase of the maximum value. But in Fig. \ref{fig:changeNkeepNA}(b), $N$ and the relative maximum of the social welfare are negatively correlated. It reflects that the increase in $N$ weakens MMZDA organizations' control over the maximum value of the social welfare, although this does not mean a decrease in the absolute maximum value of the social welfare.

\begin{figure}
\centering
    \subfigure[The absolute maximum of the social welfare]{
    \label{fig:changeNkeepNA:abs}
    \includegraphics[scale=0.252]{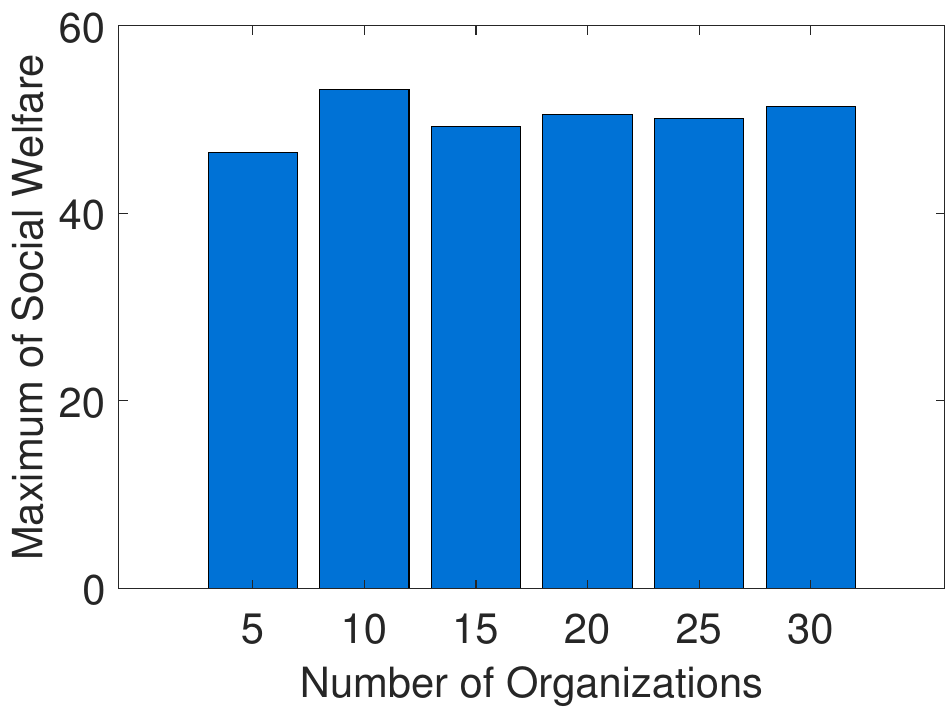}}
  \subfigure[The relative maximum of the social welfare]{
    \label{fig:changeNkeepNA:rel}
    \includegraphics[scale=0.265]{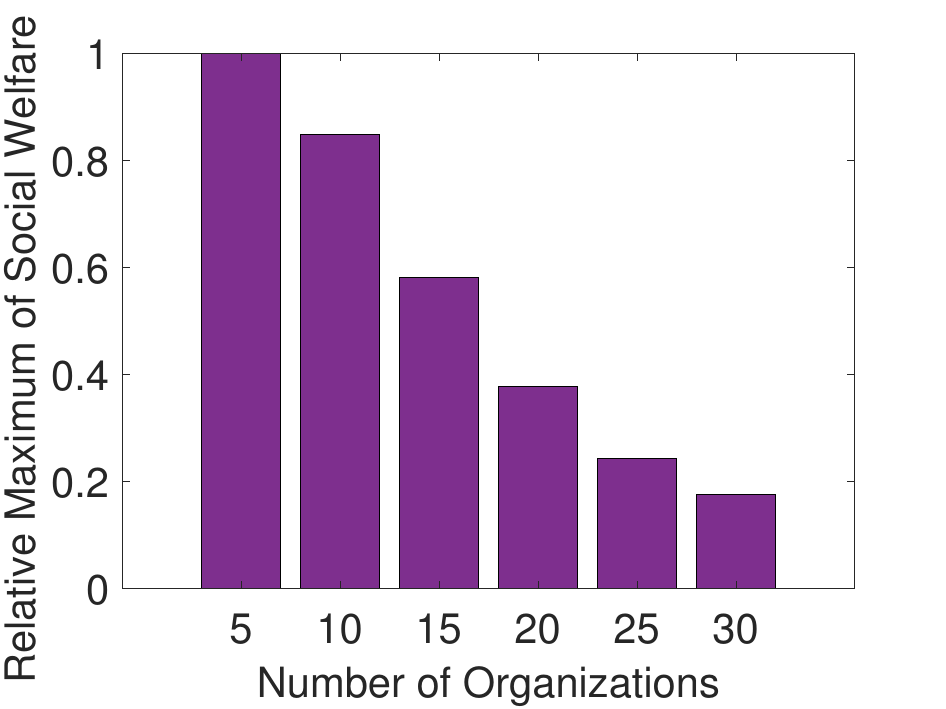}}
\caption{The absolute maximum of the social welfare and the relative maximum of the social welfare as the total number of organizations increases.}
\label{fig:changeNkeepNA} 
\end{figure}

In Fig. \ref{fig:keepNdivideNA}, we use the same dataset as the previous experiment (Fig. \ref{fig:changeNkeepNA}). In this experiment, we keep the ratio of the number of alliance organizations to the total number of organizations unchanged. Specifically, we set the number of alliance organizations to be $\frac{1}{3}$ of the total number of organizations. In Fig. \ref{fig:keepNdivideNA}(a), clearly, the maximum value of the social welfare increases as $N$ increases. This is because more organizations participate in the cross-silo FL game, and the social welfare that can be increased when the proportion of MMZDA organizations remains unchanged. While Fig. \ref{fig:keepNdivideNA}(b) implies that under the same ratio, the relative maximum of the social welfare fluctuates around $0.6$ within a certain range. There is not much change overall, and the fluctuations come from the heterogeneity of the organizations. In fact, the control ability of MMZDA also depends on the utility vectors $\mathbf{v}$ of the alliance organizations. During the experiment, we randomly select the alliance organization to more objectively reflect the expected control ability of MMZDA. 

\begin{figure}
\centering
    \subfigure[The absolute maximum of the social welfare]{
    \label{fig:keepNdivideNA:abs}
    \includegraphics[scale=0.252]{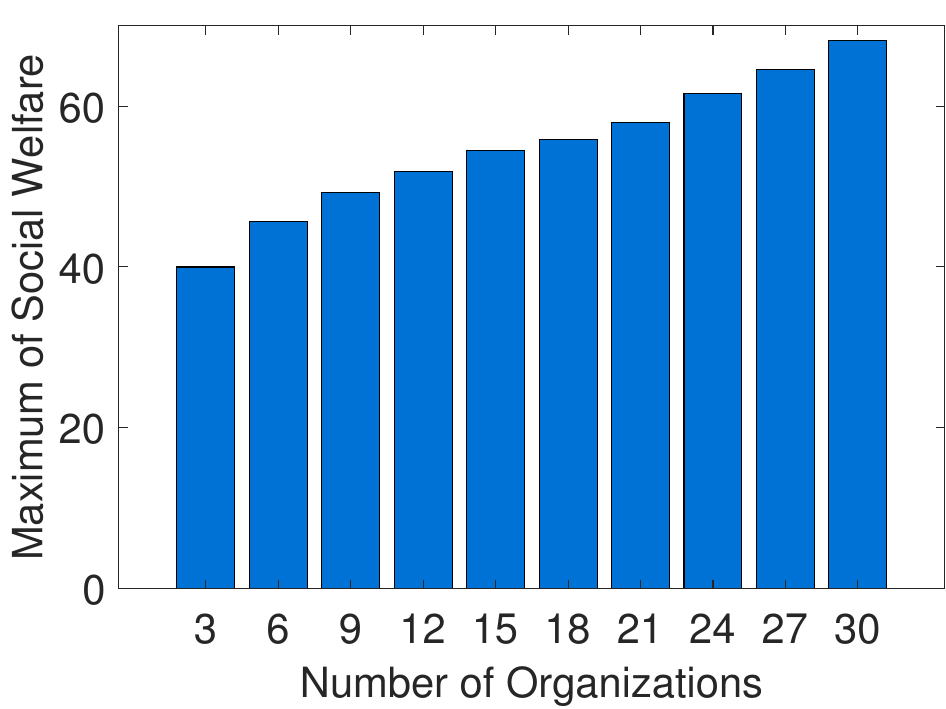}}
  \subfigure[The relative maximum of the social welfare]{
    \label{fig:keepNdivideNA:rel}
    \includegraphics[scale=0.265]{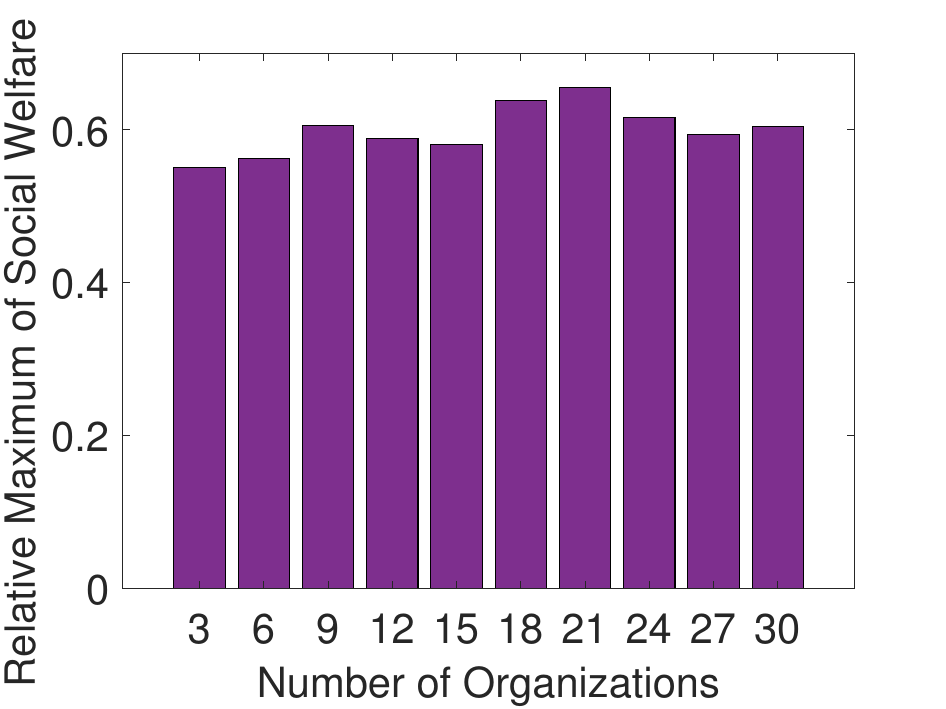}}
\caption{The absolute maximum of the social welfare and the relative maximum of the social welfare as the total number of organizations increases with the rate of MMZDA organizations unchanged.}
\label{fig:keepNdivideNA} 
\end{figure}

\section{Conclusion}
In this paper, we model the cross-silo FL game among organizations as a public goods game, revealing the social dilemma in the cross-silo FL game theoretically. In order to overcome the social dilemma, we propose a brand-new method using the MMZD to solve the social welfare maximization problem. By the means of the MMZD, an individual organization can unilaterally control the social welfare at a certain level, regardless of other organizations' strategies. Meanwhile, we explore the MMZDA consisting of multiple MMZD organizations, which further improves the control of the maximum social welfare. Moreover, our approaches can maintain the stability and sustainability of the system without extra cost. Simulation results prove that the MMZD strategy can efficiently and effectively control the social welfare. Furthermore, the MMZDA achieves a larger maximum social welfare, which shows its superiority in reducing the loss from selfish behaviors.

\ifCLASSOPTIONcaptionsoff
  \newpage
\fi



%
\bibliographystyle{IEEEtran}
\bibliography{./reference.bib,./IEEEexample}
\begin{IEEEbiography}
[{\includegraphics[width=1in,height=1.25in,clip,keepaspectratio]{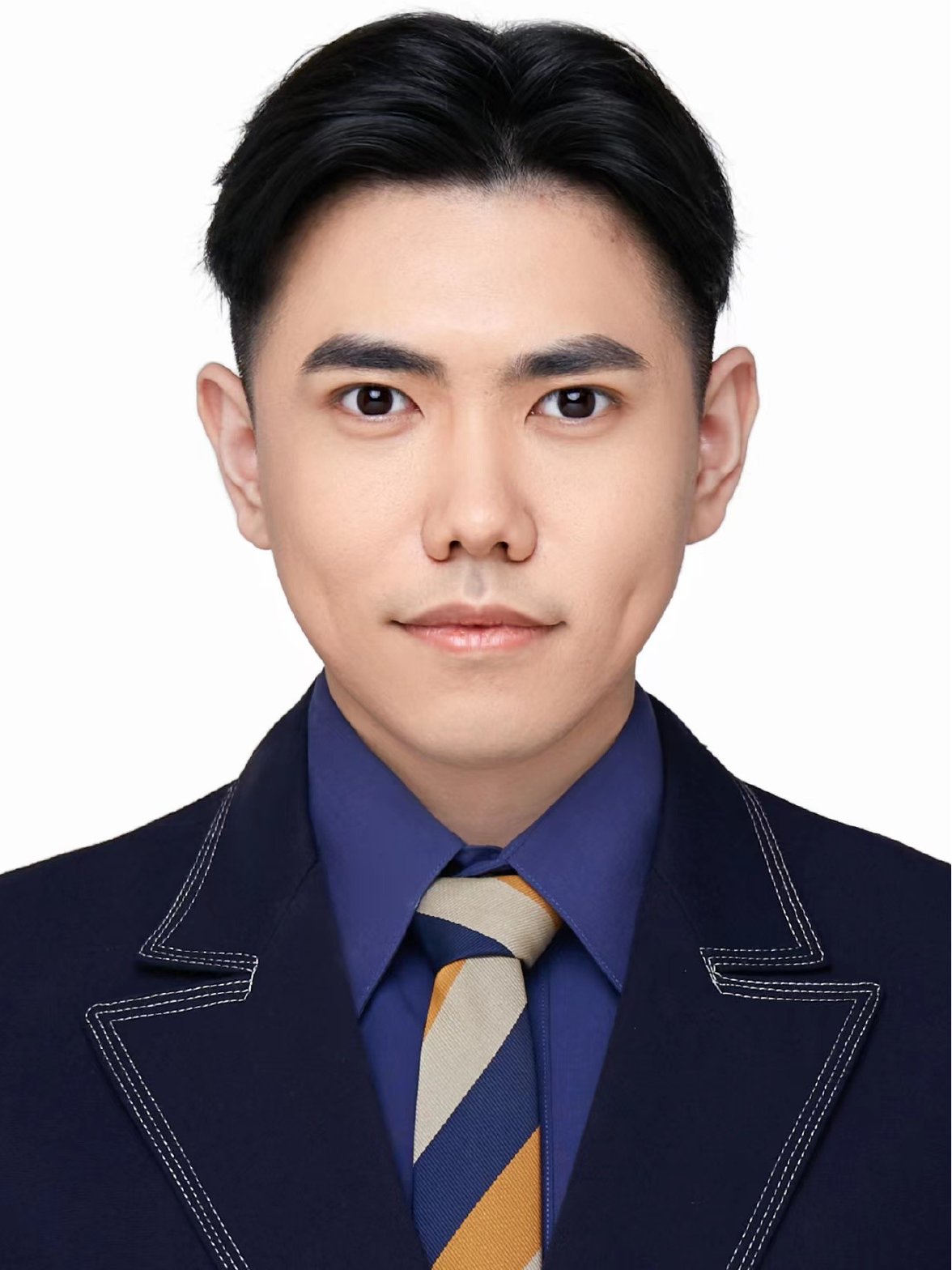}}]{Jianan Chen} received his B.S. degree in Computer Science and Technology from Beijing Normal University in 2019. He is currently pursuing his Ph.D. degree in the Department of Computer and Information Science, Indiana University-Purdue University Indianapolis (IUPUI). His research interests include block chain, network science and game theory.
\end{IEEEbiography}
\begin{IEEEbiography}
[{\includegraphics[width=1in,height=1.25in,clip,keepaspectratio]{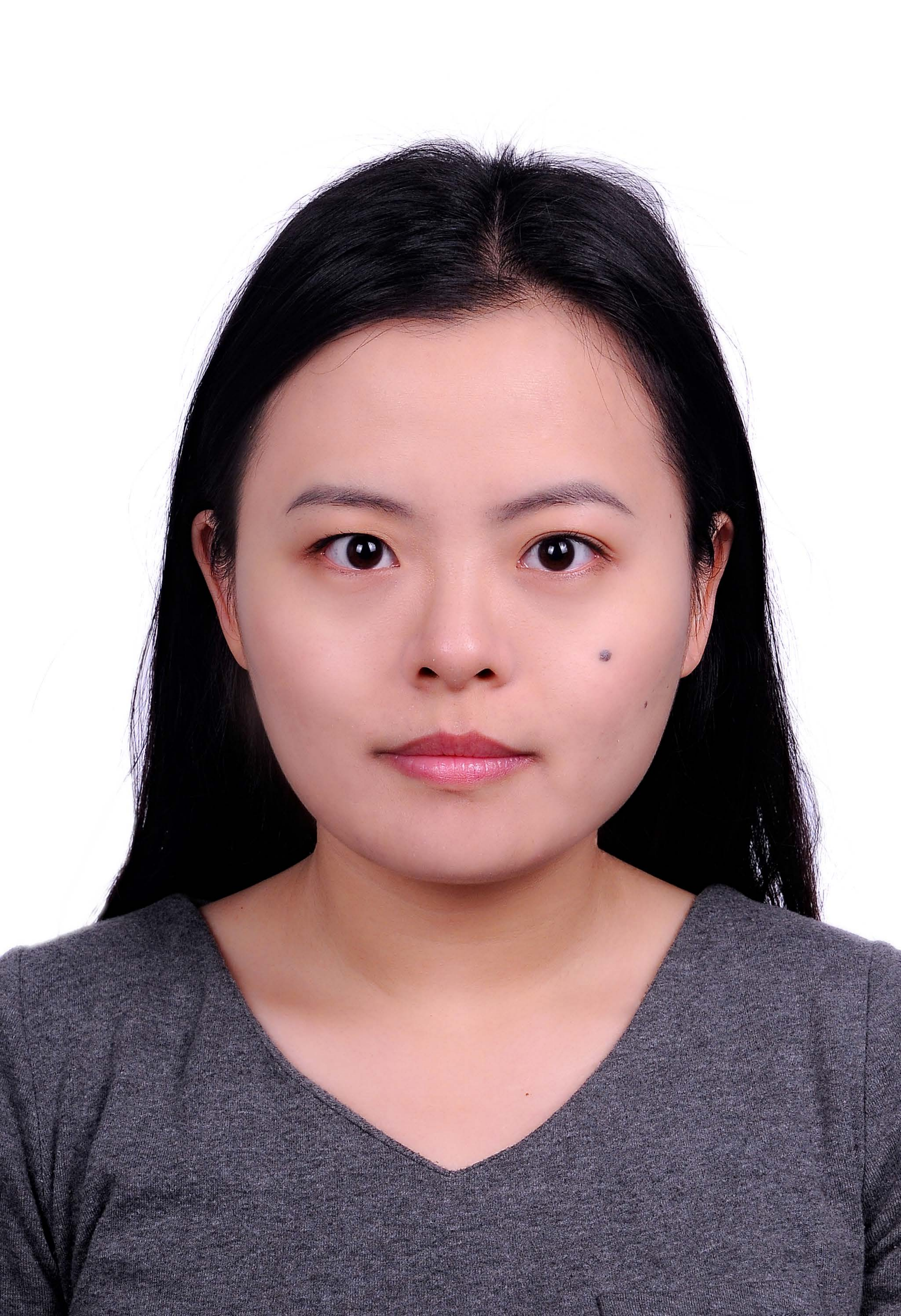}}]{Qin Hu} received her Ph.D. degree in Computer Science from the George Washington University in 2019. She is currently an Assistant Professor with the Department of Computer and Information Science, Indiana University-Purdue University Indianapolis (IUPUI). She has served on the Editorial Board of two journals, the Guest Editor for four journals, the TPC/Publicity Co-chair for several workshops, and the TPC Member for several international conferences. Her research interests include wireless and mobile security, edge computing, blockchain, federated learning, and crowdsensing.
\end{IEEEbiography}
\begin{IEEEbiography}
[{\includegraphics[width=1in,height=1.25in,clip,keepaspectratio]{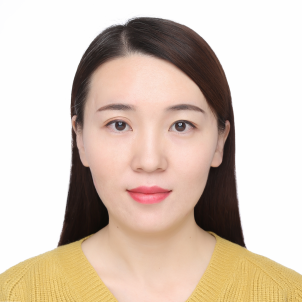}}]{Honglu Jiang} received her Ph.D. degree in computer science from The George Washington University, Washington, DC, USA, in 2021. She is currently an Assistant Professor with the Department of Computer Science and Software Engineering, Miami University, Oxford, OH, USA. Her research interests include wireless networks, differential privacy, big data, and privacy preservation.
\end{IEEEbiography}

%








\end{document}